\documentclass[runningheads]{llncs}

\usepackage{amsmath}
\usepackage{amssymb}
\usepackage{epsf}
\usepackage{graphics}
\usepackage{epsfig}
\usepackage{latexsym}
\usepackage{enumerate}
\usepackage[normalem]{ulem}
\usepackage{color}
\usepackage{float}
\usepackage{graphicx}
\usepackage[loose]{subfigure}
\usepackage{cite}
\usepackage{algorithm}
\usepackage{algorithmic}
\usepackage{epstopdf}
\usepackage{appendix}
\usepackage{placeins}

    \usepackage{changes}

\newcommand{\edge}[2]{(#1, #2)}

\begin{document}

\title{Unique Triangulated 1-Planar Graphs
\thanks{Supported by the Deutsche Forschungsgemeinschaft (DFG), grant
    Br835/20-1.}}
\author{Franz J. Brandenburg}
\institute{94030 Passau, Germany \\
           \email{brandenb@informatik.uni-passau.de}
}

\maketitle

\begin{abstract}
It is well-known that every 3-connected  planar graph has a unique planar embedding
on the sphere. We study the extension
to  \emph{triangulated 1-planar graphs}, \emph{T1P} graphs for short, which
  admit an embedding in which each edge is crossed at most once
 and  each face is a triangle, and
  obtain an algorithmic solution by  a cubic
time recognition algorithm 
that also counts the number of  \emph{T1P}  embeddings.
In particular, we show that every
triangulated planar graph has a unique \emph{T1P} embedding,
although it may admit many 1-planar embeddings,
and that any 6-connected \emph{T1P}  graph
  has a unique 1-planar embedding, except for   full generalized
two-stars  that admit  two or eight 1-planar embeddings. 
Our algorithm extends, refines, and  corrects a previous recognition
algorithm by Chen, Grigni and Papadimitiou (``Recognizing Hole-Free
4-Map Graphs in Cubic Time'', Algorithmica 45 (2006)).
\end{abstract}
\vspace{2mm}


\section{Introduction} \label{sect:intro}

Graphs are often defined by a particular property of a drawing or an
embedding in the plane or on the sphere.
 The planar graphs, in which  edges do not cross, are the
best-known and most prominent example. Planarity  is one of the most
basic and influential concepts in graph theory. Many properties of
planar graphs have been explored, including unique embeddings,
duality, minors, and straight-line drawings \cite{d-gt-00,
t-handbook-GD}.

There has been a recent interest in 
\emph{beyond-planar graphs} \cite{dlm-survey-beyond-19,
ht-beyond-book-20,klm-bib-17} which generalize   the planar graphs
by restrictions on edge crossings. Their study in graph theory,
graph algorithms, graph drawing, and computational geometry can
provide significant insights for the design of effective methods to
visualize real-world networks, which are non-planar, in general,
such as road maps and   scale-free,   social   and biological networks.
A graph is $k$-\emph{planar} \cite{pt-gdfce-97} if it admits an embedding
  such that each edge is   crossed by at most $k$ edges.
An  embedding specifies  edge crossings and faces. It is
\emph{triangulated} if each face is a triangle. A triangulated
1-planar embedding, \emph{T1P  embedding} for short, consists of
triangles and X-\emph{quadrangles}  which are quadrangles with a
pair of crossed edges in the interior, such that
there is no other vertex in the the interior  of the quadrangle.
In the plane, we also allow an X-quadrangle with
an edge crossing in the outer face. Thus there is an
 X-quadrangle if an edge is crossed. A
triangulated planar graph obtained by removing
one edge from each pair of crossed
edges. A graph is \emph{triangulated 1-planar},
\emph{T1P} for short, if it admits a \emph{T1P} embedding.
It is \emph{unique} if it has a single \emph{T1P} embedding.

The \emph{dual} of an embedded planar graph is a planar (multi-)
graph, whose vertices are  the faces of the planar embedding, such that
there is an (dual) edge if and only if the  faces  are adjacent and
share a segment. \emph{Maps}
 generalize planar duality such that some faces are ignored and
 adjacent faces can meet only in a point.  A \emph{map}
$\mathcal{M}$ is a partition of the sphere into finitely many
regions (or faces). Each region is homeomorphic to a closed disk.
Some regions are labeled as \emph{countries}  and the remaining
regions are  \emph{holes} of $\mathcal{M}$. In the plane, we assume
that the outer region is a country.
   Two countries are \emph{adjacent} if
their boundaries  intersect. Then they meet a  segment  or  in
a single point, in general. There is a
$k$-\emph{point} $p$ if $k$ countries meet   in $p$. A map
$\mathcal{M}$ defines a graph $G$, whose vertices are the countries
of $\mathcal{M}$   and there is an edge if and only if the countries
are adjacent. Then $G$ is called a \emph{map graph} and
$\mathcal{M}$ is a map of $G$. A map is a $k$\emph{-map} if no
more than $k$ countries meet in a point, and it is \emph{hole-free}
if all regions are countries. A graph is a \emph{hole-free}
$k$-\emph{map graph} if it is the map graph of a hole-free $k$-map.

Maps and map graphs have been introduced by Chen et al.
\cite{cgp-pmg-98, cgp-mg-02, cgp-rh4mg-06}.
 They  observed that the planar graphs are the 2-map or the 3-map graphs,
and that a graph is a  3-connected hole-free 4-map graph if and only
if it is  triangulated and 1-planar. Hence, a \emph{T1P} graph can
be represented  by a 1-planar embedding or by a hole-free 4-map  and also
by a bipartite planar  witness  as described in \cite{cgp-mg-02}. A
\emph{T1P} graph with $n$ vertices has at least $3n-6$
and at most $4n-8$ edges, where the
upper bound is tight for $n=8$ and all $n \geq 10$ \cite{bsw-1og-84}. Such
graphs are called \emph{optimal 1-planar}. They  have a unique
1-planar embedding  \cite{s-s1pg-86,s-rm1pg-10},
except for  full generalized two-stars, as shown in
Fig.~\ref{fig:XW},   which have two
1-planar embeddings for $n\geq 10 $ and eight for $n=8$,
that  are also \emph{T1P} embeddings.

The \emph{complete graph} on four vertices $K_4$ is the building
block for a \emph{T1P} graph. It has two representations, namely
 as a ``tetrahedron'' and an ``X-quadrangle'' in a 1-planar embedding or as a  ``rice-ball'' and a  ``pizza'' in a map, see Fig.~\ref{fig:k4-embeddings}.
In the plane, a \emph{tetrahedron}
 consists of a 3-cycle as its outer boundary and a center in its interior.
  The edges of a tetrahedron can be crossed by other  edges,
see Figs.~\ref{fig:kite-covered-tetrahedron} and \ref{fig:SC-graph}.
 A $K_4$ with vertices  $a,b,c,d$ has three embeddings as an X-quadrangle, such that one of the  edges $\edge{a}{b}, \edge{a}{c}$ or $\edge{a}{d}$ is crossed,
and the crossing edge is independent.
Hence, a $K_4$ is not a unique \emph{T1P} graph. However, may admit only a single
1-planar embedding if it is a part of a larger 1-planar graph.

\begin{figure}  
  \centering
  \subfigure[]{
    \includegraphics[scale=0.45]{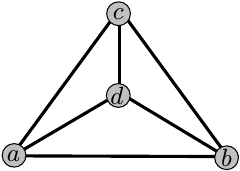}
    \label{fig:tetrahedron}
  }
  \hfil
  \subfigure[]{
      \includegraphics[scale=0.4]{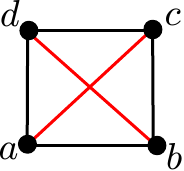}
      \label{fig:kite}
  }
    \hfil
   \subfigure[] {
      \includegraphics[scale=0.50]{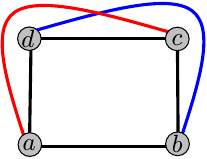}
      \label{fig:Wconf}
  }
  \hfil
   \subfigure[] {
      \includegraphics[scale=0.50]{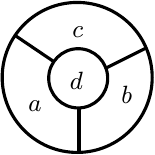}
      \label{fig:maptetra}
  }
  \hfil
   \subfigure[] {
      \includegraphics[scale=0.50]{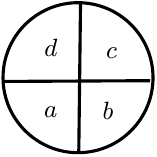}
      \label{fig:pizza}
  }
  \hfil
   \subfigure[] {
      \includegraphics[scale=0.50]{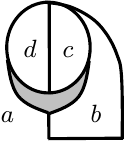}
      \label{fig:outpizza}
  }
  \caption{Representations of $K_4$  as (a)  a
    tetrahedron and   an X-quadrangle with edges crossing (b) in the inner
    and (c) in the  outer face in a 1-planar embedding,   and
    (d) a rice-ball and (e,f) a pizza in a   map.
    }
  \label{fig:k4-embeddings}
\end{figure}

Recognizing 1-planar graphs is NP-complete
\cite{GB-AGEFCE-07,km-mo1ih-13},   in general, even if further restrictions are
imposed \cite{abgr-1prs-15, bbhnr-NIC-17, bdeklm-IC-16}. On the
other hand, there is a linear time algorithm for the recognition of 1-planar
graphs if they have  the maximum of $4n-8$ edges \cite{b-ro1plt-18}.
Triangulated 1-planar graphs are in between.
  Chen et al.~\cite{cgp-rh4mg-06} have described  a cubic time recognition
algorithm for hole-free 4-map graphs. They have used maps in for  the
presentation. 
  However, their  proof is not correct and their algorithm may
  fail on  some hole-free 4-map graphs, that is it is   false
  positive, as we will show in Section~\ref{sect:quadruple}.
 We fix their  bugs after a  detailed case analysis and  distinguish
between unique and ambiguous generalized separators. This gives more
flexibility for the order in which generalized separators are used.
Brandenburg \cite{b-4mapGraphs-19} has used the previous algorithm
for an extension to 4-map graphs with holes. He has reduced the
correctness of his algorithm to the correctness of the algorithm by
Chen et al.~\cite{cgp-rh4mg-06}.
 For general map
graphs, Thorup \cite{t-mgpt-98} has claimed a polynomial time
recognition algorithm. However, there is no description of the
algorithm, such that the recognition problem for map graphs is still
open. There are more results on the recognition of map graphs if the
graphs have bounded girth \cite{LeLe-map-girth-19}, bounded
treewidth \cite{abdgmt-map-treewidth-22} or an outerplanar witness
\cite{mrs-outermap-18}. \\

\noindent \textbf{Our contribution.} We extend, refine and correct  the cubic time
recognition algorithm for \emph{T1P} or hole-free 4-map graphs by
Chen et al.~\cite{cgp-rh4mg-06}, such that it clearly distinguishes
between unique and ambiguous generalized separators.
In addition, it computes the
number of \emph{T1P} embeddings. We use 1-planar embeddings for the
  correctness
proof, which seems simpler than using maps, in particular for separating triangles.
Our algorithm  returns a \emph{T1P} embedding as a witness
for a  \emph{T1P} graph. It provides an algorithmic
solution for the question whether or not a  \emph{T1P} graph is unique.
 In particular, we show that every
triangulated planar graphs has a  unique   \emph{T1P} embedding,
where one crossing per edge does not allow more embeddings,
and similarly for 6-connected  \emph{T1P} graph, where only full generalized two-stars
have two or eight 1-planar  embeddings.

Our main result is as follows.

\begin{theorem}  \label{thm1}
There is a cubic-time recognition algorithm that decides whether or
not a graph  is  a  T1P  graph  (or a hole-free 4-map graph). If
positive, the algorithm returns a T1P embedding and
computes the number of  T1P
 embeddings.
\end{theorem}

 \noindent \textbf{Paper organization.}   We introduce basic notions 
 in the next section,
    describe our  recognition algorithm in  Section~\ref{sect:algorithm}, study
    applications in Section~\ref{sect:app}, and
  conclude in Section~\ref{sect:conclusion}.

\section{Notation and Basic Notions} \label{sect:basics}
 We consider graphs that are \emph{simple} both in a graph theoretic
and in a topological sense. Thus there are no multi-edges or loops,
adjacent edges do not cross, and two edges cross at most once in an
embedding. A graph $G=(V,E)$ consists of   sets of $n$ vertices and
$m$ edges, such that $n$ is the \emph{order} of $G$. It is called
\emph{nontrivial} if its has at least two vertices and an edge. The set of \emph{neighbors} of
a vertex $v$ is denoted by $N(v)$, and  the
\emph{subgraph} of $G$ induced by a subset $U$ of vertices
    by $G[U]$. Let $G_1+G_2$ denote $G[U_1
\cup U_2]$ if $G_1$ and $G_2$ are subgraphs of $G$ with sets of
vertices $U_1$ and $U_2$, respectively, and similarly let $G-H$
denote $G[V-U]$ if $U \subseteq V$ and $U$ is the set of vertices of
$H$. If $F$ is a set of edges, then $G-F$ is the subgraph of $G$
with the edges of $F$ removed. For convenience, we omit braces if
the meaning is clear, for example, for singleton sets,
such that, for example,  $K_5$-$e$ denotes the 5-clique with one edge removed.

We assume that a graph is defined by an  \emph{embedding}
$\Gamma(G)$ on the sphere or in the plane, which is 1-\emph{planar} if
each edge is crossed at most once, and
  \emph{triangulated} if all faces are triangles.
  In the plane, the outer face is a  triangle with
    three vertices or two vertices and a crossing point
  if the embedding is 1-planar and triangulated.
Obviously, a 1-planar  embedding   is  \emph{triangulated}
if it consists of triangles and X-quadrangles.  Obviously, there are 1-planar
embeddings and 1-planar graphs that are not triangulated. In between
are \emph{kite-augmented 1-planar embeddings} (and graphs)
in which the subgraph, induced by the vertices of any pair of
crossed edges, is a 4-clique  that is
represented as an X-quadrangle, see Fig.~\ref{fig:small-planar}.
Now an embedding
admits large faces, and the corresponding map graphs may have holes
 \cite{b-4mapGraphs-19}.

\begin{figure}[t]  
  \centering
\subfigure[\emph{T1P}] {
     \includegraphics[scale=0.55]{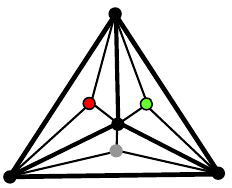}  
    \label{fig:small-planar}
  }
\hspace{3mm}
\subfigure[kite-augmented] {
     \includegraphics[scale=0.55]{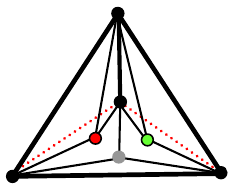}  
    \label{fig:small-kite}
  }
\hspace{3mm}
 \subfigure[1-planar] {
     \includegraphics[scale=0.55]{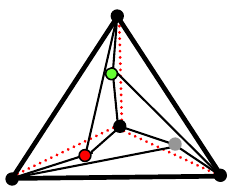}  
     \label{fig:small-1planar}
   }
  \caption{A triangulated planar   graph with three 1-planar embeddings.
  }
  \label{fig:small-planar}
\end{figure}

A graph $G$ is called a  \emph{T1P} \emph{graph} if it admits a
triangulated 1-planar embedding. A graph is \emph{unique} for its class
(planar, \emph{T1P}, 1-planar)
 if it has a single  embedding,
 up to a topological
homomorphism~\cite{gy-graphtheory-99},
that fulfills the requirements of the class.
 Clearly, \emph{T1P} graphs are
3-connected. Even more, the planar subgraph, obtained by removing
all crossed edges from an embedding $\Gamma(G)$, is 3-connected
for $n \geq 5$
\cite{abk-sld3c-13}. As observed by Chen et al.~\cite{cgp-mg-02,
cgp-rh4mg-06}, a graph is a \emph{T1P} graph if and only if it is a
3-connected hole-free 4-map graph. Thus we can use
\emph{T1P} graph and 3-connected 4-map graph
interchangeably.

 In general, a graph has many embeddings of a particular type,
e.g., planar or  1-planar ones, see graph $K_5$-$e$ in
Fig.~\ref{fig:small-planar} and
Figs.~\ref{fig:K5-e-T1P} - \ref{fig:K5-e-ill}.
 A 2-connected planar graph can have
exponentially many planar embeddings, which can   be
stored in a
 PQ-tree \cite{bl-tc1-76} and an SPQR-tree \cite{dt-lmtc-96}, respectively.
 Bachmaier et al.~\cite{bbhnr-NIC-17} have shown that
there are  maximal \emph{NIC}-planar graphs with an exponential number of
\emph{NIC}-planar embeddings. Their graphs   are composed of about $\frac{n}{5}$
independent 5-cliques. These graph \emph{T1P} graphs and
  have an exponential number of \emph{T1P} embeddings,
as we will show.
On the other hand,
  a planar graph is  unique if it is 3-connected
\cite{w-uniqueplanar-33}.

We will use the following notions for \emph{T1P} graphs and their
embeddings. A $k$-\emph{cycle} $C(v_1,\ldots, v_k)$ for $k \geq 3$ is
a sequence of vertices such that there is an edge
$\edge{v_i}{v_{i+1}}$ for $i=1,\ldots,k$, where $v_{k+1}=v_1$. It is
an \emph{induced cycle} if it has no chord. A $k$-clique
$K(v_1,\ldots, v_k)$ is a set of $k$ pairwise adjacent vertices. It
is \emph{maximal} if it not properly contained in another  clique.
This is assumed from now on, in particular, for 4-cliques.
 An   X-\emph{quadrangle} $Q(a,b,c,d)$ is an embedding of a 4-clique,
 such that edges $\edge{a}{c}$ and $\edge{b}{d}$ cross and vertices $a,b,c,d$ appear in this
 order in the boundary, see Figs.~\ref{fig:kite} and
\ref{fig:Wconf}. Thus $Q(a,b,c,d)$ and $Q(a,c,b,d)$ are different.
We assume that there are no other vertices in the
interior (or in the exterior) of the boundary of
an X-quadrangle. Otherwise, the 1-planar embedding is
not triangulated.

\begin{figure}[t]  
  \centering
    \includegraphics[scale=0.5]{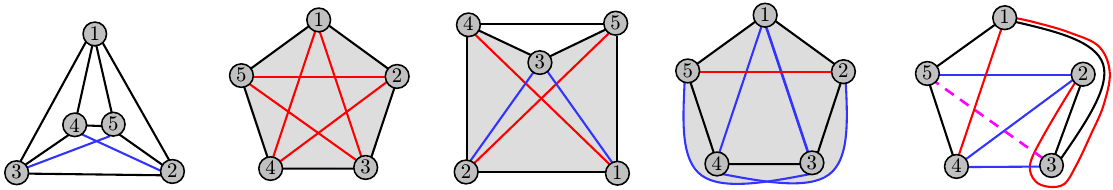}
 \caption{All topological embeddings of $K_5$ \cite{hm-dcgmnc-92}.  }
  \label{fig:allK5}
\end{figure}

\begin{figure}[b]  
  \centering
\subfigure[]{
    \includegraphics[scale=0.55]{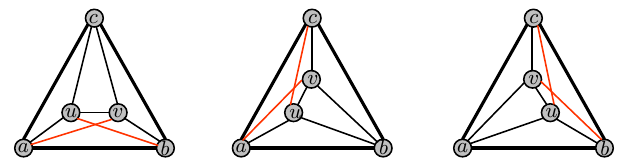}  
\label{fig:K5variation}
}
\hspace{3mm}
\subfigure[]{
    \includegraphics[scale=0.3]{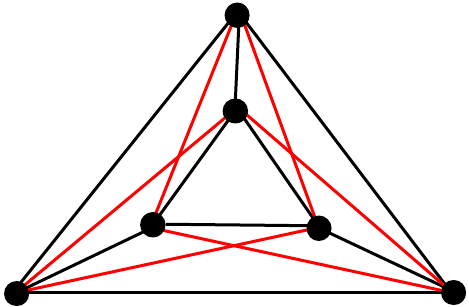}
    \label{fig:K6-triangle}
    }
  \caption{(a) Three 1-planar embeddings of $K_5$ with a fixed outer face.
Vertices $u$
  and $v$   change roles for the remaining three 1-planar embeddings,
and
(b) the 1-planar embedding  of $K_6$.
}
  \label{embeddingsK5andK6}
\end{figure}

 As observed before, a 4-clique  has three 1-planar embeddings as an
 X-quadrangle and one as a  tetrahedron.
 A (maximal) 5-clique  has five topological embeddings, as shown in
 Fig.~\ref{fig:allK5}, one of which is
 1-planar   \cite{hm-dcgmnc-92}. In the plane, the 1-planar embedding of $K_5$
 has  an  outer triangle  $C(a,b,c)$
 and two  inner vertices  $x$ and $y$ in the interior of $C(a,b,c)$.
The outer triangle is chosen such that there is another vertex
outside $C(a,b,c)$ if the 5-clique is a part of a larger graph.
  A 5-clique has five $K_4$, one of which is represented as an
 X-quadrangle in a 1-planar embedding, see Fig.~\ref{fig:allK5}. It admits  six 1-planar embeddings, each with a different pair of crossed edges and a triangulation.
There are less embeddings if some edges must be uncrossed.
Note that a 6-clique has a unique
1-planar embedding,  which is shown in Fig.~\ref{fig:K6-triangle}. It is not a
subgraph of a 4-connected 1-planar graphs of order at least seven
\cite{cgp-rh4mg-06}.

At last, we consider  five variations of \emph{generalized two-stars} $G2S_k$,
  see Fig.~\ref{fig:G2S}. First, there is a sequence
of vertices $v_1,\ldots, v_k$ for $k \geq 5$. Each vertex $v_i$ is
adjacent to its neighbors and its neighbors after next for
$i=1,\ldots,k$.
  An edge $\edge{v_i}{v_{i+2}}$ is called an \emph{arch}.
There are two \emph{poles} $p$ and $q$, which are each
a neighbor of    $v_i$ for $i=1,\ldots,k$.
So far, vertices $v_1$ and $v_k$ have degree four, $v_2$
and $v_{k-1}$ have degree five, the vertices in between have
degree six, and the poles have degree $k$.
Generalized two-stars are augmented for a triangulation,
see Fig.~\ref{fig:DP}.
 We call edge $\edge{p}{q}$
 a \emph{handle} and obtain  graph  $hG2S_k$ if it added to   $G2S_k$.
Graph $cG2S_k$ is obtained if edge $\edge{v_1}{v_k}$ is added to   $G2S_k$, such that
there is a circle for $v_1,\ldots, v_k$. There is  $xG2S_k$ if
both $\edge{p}{q}$ and $\edge{v_1}{v_k}$ are added. At last, if $k$
is even, then edges $\edge{v_1}{v_k}$ and $\edge{v_2}{v_k}$ can be added
to obtain a semi-full generalized two-star $sG2S_k$. Finally, there is  a
full generalized two-star $fG2S_k$ if  edges $\edge{v_1}{v_k},
\edge{v_1}{v_{k-1}}$ and edge $\edge{v_2}{v_k}$ are added to
$G2S_k$, see Fig.~\ref{fig:XW}.
  Then there is a   cycle with arches in circular order for
  $v_1,\ldots, v_k$.
Graph $fG2S_k$ has been introduced by Bodendiek
et al.\cite{bsw-1og-84} as   ``double wheel graph''    $2*\hat{C}_{2n}$.
Generalized two stars have (at least) two 1-planar embeddings, in which the poles $p$ and $q$ change roles, and so do incident edges.
 The embeddings  coincide up to
graph isomorphism, such that there is the same picture. Graph
$fG2S_6$  consists of a planar drawing of a cube with a pair of crossed edges in each face.
It has eight 1-planar embeddings \cite{s-rm1pg-10}.

\begin{figure}[t]   
  \centering
  \subfigure[]{
    \includegraphics[scale=0.5]{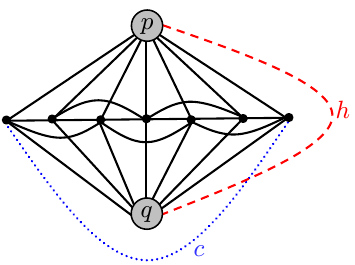}     
    \label{fig:DP}
    }
    \hspace{5mm}
 \subfigure[]{
        \includegraphics[scale=0.5]{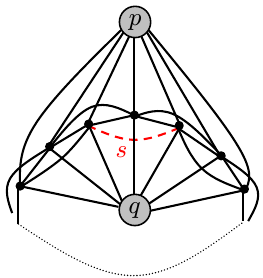}
\label{fig:XW}
     }
 \caption{(a) A generalized two-star $G2S_k$ augmented by a handle $h$, drawn
 red and  dashed, and augmented to a circle by $c$, drawn blue and dotted.
 There is a $xG2S_k$ if $h$ and $c$ are added
 and (b)  a full generalized two-star, which is semi-full if edge
 $s$ is missing.
  }
  \label{fig:G2S}
\end{figure}

\section{Generalized Separators}  \label{sect:algorithm}

We extend, refine and  correct   the recognition algorithm
for hole-free 4-map graphs by Chen et al.~\cite{cgp-rh4mg-06},
called algorithm $\mathcal{A}$.   It uses
several generalized separators  that partition a graph into parts.
An edge is called a \emph{bridge} if its vertices are in different parts
from a partition.
Algorithm $\mathcal{A}$ distinguishes between (maximal)
4-cliques and 5-cliques and considers the latter after separating
triangles.
The algorithm is complex and has a long proof.
Parts of the proof are not described in the
paper \cite{cgp-rh4mg-06} and are deferred to a technical report,
(upon a request by the editor and the reviewers, as Z. Chen
has communicated).
 But it has a bug and  a   gap, which are due to ambiguities.
In Lemma 7.9, Chen et al.~\cite{cgp-rh4mg-06} claim that a separating
quadruple, which is studied in Section~\ref{sect:quadruple}, has
exactly one bridge. That's  wrong, as
Figs.~\ref{fig:quad-many-bridges} and \ref{fig:quadrangles} show.
Second, a  bridge may not cross its defining edge as claimed.
At last, a
 bridge may have two options for a crossing, as shown in Fig.~\ref{fig:quad-ambig},
such that there
is an ambiguous separating quadrangle, as opposed to
Lemma 7.9 in \cite{cgp-rh4mg-06}.  There is a similar problem  for
strongly separating triangles in Lemma 7.14, 
since an edge and its bridge do not necessarily cross, as we will show in
Section~\ref{sect:stronglyseparatingtriangles}.\\

Forthcoming, we assume that $G$ is a \emph{T1P} graph  whose edges
are marked and labeled.
 Initially,   all edges are
\emph{unmarked} and $\lambda(e)=1$ for the \emph{label} of each edge $e$.
Later on, graph $G$ is obtained from an \emph{input
graph} by the removal of crossed edges and a partition at separating $k$-cycles using our
algorithm, called (new) \emph{algorithm} $\mathcal{N}$.
 An edge is \emph{marked} if it is uncrossed
in every 1-planar embedding of (the actual) $G$. A marked edge may be crossed
in the input graph if a crossing has been detected at an earlier
stage and the  crossing edge has been removed.
  Label $\lambda(e)$ of   edge $e$ is an integer that counts the number of
\emph{T1P} embeddings for $e$.
Label $\lambda(e)$ is changed if edge $e$ is part of two or more  \emph{T1P}
embeddings. The total number of  \emph{T1P} embeddings of $G$ is the
product of the labels of all edges.

A \emph{generalized separator} consists of an induced $k$-cycle $S$
and a set of $\ell$ edges $T$ for $2 \leq k \leq 5$ and $0 \leq \ell \leq 3$,
such that $G'= G-S-T$ is disconnected.
It  is   \emph{strong} if
$G'$ consists of two   parts, one of which is a single
vertex. An edge $\hat{e} \in T$ is called a \emph{bridge} if its vertices
are in different parts of $G'$, where two parts may be connected
by several bridges, as opposed to some assumptions in \cite{cgp-rh4mg-06}.
For an efficient computation,   generalized separators are further
restricted.    The set of \emph{crossable edges}
$\mathcal{E}\edge{u}{v}$ of an unmarked edge $\edge{u}{v}$ consists
of all unmarked edges $\edge{x}{y}$ such that $G[u,v,x,y]$ is a
maximal 4-clique. Let $\mathcal{E}(e)=\emptyset$ if $e$ is marked.
Let $\mathcal{B}(e) \subseteq \mathcal{E}(e)$ denote the set  of
bridges of edge $e$. It is computed by a connectivity test \cite{clrs-ia-01}.
Each bridge $\hat{e} \in \mathcal{B}(e)$ is called
\emph{a bridge over} $e$. In general, edge $e$ has a single bridge $\hat{e}$,
such that
$e$ and $\hat{e}$ cross in a \emph{T1P} embedding, except if there is an ambiguity.

The set of pairs $(C, \mathcal{E}(e))$, where $C$ is an induced
$k$-cycle  ($k \leq 5$) and   $e$ is an edge of $C$,  can be computed
in linear time. This is due to Chiba and Nishizeki
\cite{cn-asla-85}, who have shown that all $k$-cycles and all maximal
$k$-cliques can be computed in linear time if the graph has bounded
arboricity and $k$ is constant. Hence, all $k$-cycles (with and
without a chord) for $k=3,4,5$ and all 4- and 5-cliques (that are assumed
to be maximal) of a
\emph{T1P} graph can be computed in linear time.
Note that  a 4-connected 1-planar  graph does not
contain $K_6$ as a proper subgraph \cite{cgp-rh4mg-06},
and that $K_7$ is not 1-planar.
\\

Like its predecessors  \cite{cgp-rh4mg-06, b-4mapGraphs-19},
 algorithm $\mathcal{N}$ ``makes progress'' and either partitions
  $G$ into smaller components using a  separating $k$-cycle or it finds
  edges  that are crossed in every \emph{T1P}  embedding and removes
them towards planarity.
In addition, it determines ambiguities and updates the label of edges.

The outline of  algorithm $\mathcal{N}$ 
is as follows:

\begin{enumerate}[(1)]
\item Preprocess $G$ and check that it is 3-connected and has between
$3n-6$ and $4n-8$ edges
and compute all induced $k$-cycles for $k\leq 5$ and all (maximal) 4-  or   5-cliques.
\item 
Mark
all edges   that are not part of a   4-  or   5-clique.
\item Compute all pairs $(e, \mathcal{E}(e))$ of unmarked
edges  and their set of crossable edges.
\item Apply a generalized separator in the order described below.
\item  Destroy all remaining X-quadrangles.
\item Check planarity.
\item Compute a \emph{T1P} embedding.
\end{enumerate}

The algorithm will stop if there is a failure at any step. Its core is step (4),
which is described below in all detail. We will
  state the properties of each generalized separator and
prove them using 1-planar embeddings.
This part is by recursion, such that
there is an update for the set  of  induced $k$-cycles
and of pairs  $(e, \mathcal{E}(e))$, such that 3-cycles
are  search next if
any generalized separator has been applied successfully.
After step (4), graph  $G$ has none of the generalized separators.
Then either $G$ is \emph{small} and of order at most six,
or each remaining  4-clique $K(a,b,c,d)$ with a pair of
 unmarked edges $\edge{a}{b}$ and $\edge{c}{d}$ is represented as an
X-quadrangle, 
whereas the edges of a tetrahedron  are marked, such
that none of its edges has a crossable edge. Hence, only the cyclic
order of the vertices or the pair of crossed edges must be
determined. This is done by (simultaneously) replacing each
remaining 4-clique $K(a,b,c,d)$  that must be represented as an
X-quadrangle by a 4-star with a new vertex $z$ that is connected to
each of $a,b,c,d$. An edge $\edge{a}{b}$ of $K(a,b,c,d)$ is kept if
it is marked and is removed, otherwise. The cyclic order of
$a,b,c,d$ is determined by the embedding from the planarity test
in step (6). Alternatively, one may search for edges that are shared by
two 4-cliques as in \cite{cgp-rh4mg-06}.
For step (7),  starting from
the planar embedding after step (6),   all edges are reinserted
that were removed in step (4) and embeddings are composed
at  separating $k$-cycles, such that the uncrossed
edges from the $k$-cycles are identified.
For an algorithmic description we refer to \cite{b-4mapGraphs-19}.

Algorithms $\mathcal{N}$ and
$\mathcal{A}$ differ in some details. First, Chen et
al.~\cite{cgp-rh4mg-06} treat 4-cliques and 5-cliques
separately. They  search for $K_5$ only in 5-connected or small graphs. Their
search for the X-quadrangle in an embedding of a 5-clique is quite
involved. We show that   5-cliques are special and can be classified
by the degree of the vertices using $(3,\ell)$ separators for
$\ell=0,1,2$. Moreover, we distinguish between unique and
other separating triples and quadruples and show that
 unique separating quadrangles, bridge-independent separating
triangles, and singular separating triangles can be searched in any order, whereas
separating triangles must be searched last in \cite{cgp-rh4mg-06},
as (the proof of) Lemma 7.15 shows. We
introduce ambiguous generalized
separators and   use tripods. In consequence, our small graphs
have order at most six, whereas they have up to eight vertices
 in~\cite{cgp-rh4mg-06}.
Finally, there is a planarity test,
where we prepare  the obtained graph
in a special way.

We use the following generalized separators in the listed order,
where some changes are possible. A
generalized separator is marked by an asterisk (*) if it has not
been used by Chen et al.~\cite{cgp-rh4mg-06}.

\begin{figure}[t]
  \centering
  \subfigure[separating edge] {
     \includegraphics[scale=0.35]{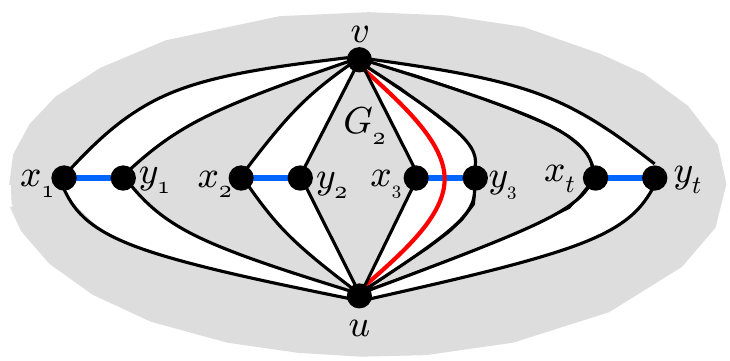}
    \label{fig:separatingedge1}
  }
  \subfigure[separat. triple] {
     \includegraphics[scale=0.4]{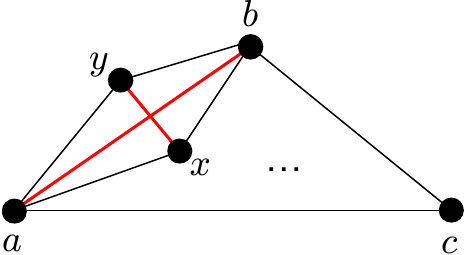}
    \label{fig:triple}
  }
  \subfigure[separ. triangle] {
      \includegraphics[scale=0.4]{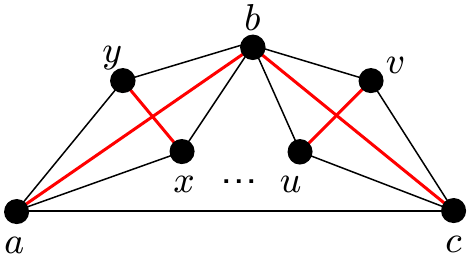}
      \label{fig:triangle}
  }
   \caption{Generalized separators.
}
  \label{fig:separatingtriangles}
\end{figure}

\begin{definition} \label{def:separators}
Let $G$ be a 3-connected  graph with marked (and labeled) edges.
\begin{enumerate}
\item A   3-cycle $C(a,b,c)$ is \emph{separating} if
   $G'=G-C(a,b,c)$ is disconnected.
  \item  Edge $\edge{u}{v}$ is a \emph{separating edge}  if $G'=G -\{u,v\}-\mathcal{E}\edge{u}{v}$ is disconnected, see Fig.~\ref{fig:separatingedge1}.
  \item  A \emph{separating triple}  consists of a 3-cycle
  $C(a,b,c)$ such that $G'=G-C(a,b,c)-\mathcal{E}\edge{a}{b}$ is disconnected,  see
  Fig.~\ref{fig:triple}.
 \item (*)  A  \emph{$K_5$-triple}
  is a 5-clique with a vertex  of degree four.
 \item A \emph{separating 4-cycle}  is an induced 4-cycle $C(a,b,c,d)$
such that  $G'=G-C(a,b,c,d)$ is disconnected.
\item A   \emph{separating quadruple} consists of an induced 4-cycle $C(a,b,c,d)$
such that  $G'=G-C(a,b,c,d)-\mathcal{E}\edge{a}{b}$ is disconnected.\\
(*) It is \emph{unique} if  $G'$ has a single bridge and  $G-C(a,b,c,d)-
\mathcal{E}(e)$ does not decompose for an  edge $e \neq \edge{a}{b}$ of $C(a,b,c,d)$.
%
%
  \item  A \emph{separating triangle} consists of a 3-cycle  $C(a,b,c)$,
  such that $G'=G-C(a,b,c)-\mathcal{E}\edge{a}{b}-\mathcal{E}\edge{b}{c}$
  is disconnected, see Fig.~\ref{fig:triangle}. \\
(*) It is \emph{bridge-independent} if  $G'$ has two nontrivial parts that
are connected by independent bridges.
\item  (*) A separating triangle is \emph{singular}  if it is strong, such that one part of
   $G'$  is a  vertex $d$, there are two
   bridges $\edge{d}{x}$ and $\edge{d}{y}$   and there is no edge
   $\edge{x}{y}$.
 \item  (*) A
    $K_5$-\emph{destroyer} is
  a 5-clique $K(a,b,c,x,y)$ of $G$ with vertices $x,y$ of degree
  five, a 5-clique $K(a,b,c,y,z)$ and 4-cliques $K(a,c,x,u)$ and
  $K(a,b,z,v)$.
\item (*) An \emph{ambiguous separating quadruple}
consists of a strong separating quadruple with a bridge  $\edge{x}{y}
  \in \mathcal{E}\edge{a}{b} \cap \mathcal{E}\edge{b}{c}$.
 \item (*) An \emph{ambiguous separating triangle}  is a strong
  separating triangle such that  there   are bridges $\edge{d}{x}$ and $\edge{d}{y}$  and there is an edge
  $\edge{x}{y}$, see Fig.~\ref{fig:ambigtriangle}(a) and (b).
%
  %
%
 \item A \emph{separating 5-cycle}  is an induced 5-cycle $C(a,b,c,d,e)$
such that  $G'=G-C(a,b,c,d,e)$ is disconnected.
\item  (*) A \emph{separating tripod} consists of a 3-cycle $C(a,b,c)$
such that
$G'=G-C(a,b,c)-\mathcal{E}\edge{a}{b}-\mathcal{E}\edge{b}{c}-
\mathcal{E}\edge{a}{c}$ is disconnected, see Fig.~\ref{fig:kite-covered-tetrahedron}.
\item (*) A \emph{singleton separating tripod} is a strong separating tripod $C(v_1, v_2, v_3)$
 such that $G'$ decomposes into two parts, one of which  is
    a single vertex $u$ (of degree six), there are three bridges
    $\edge{u}{v_i} \in \mathcal{E}\edge{v_i}{v_{i+1}}$ for $i=1,2,3$
    with $v_4=v_1$ and at most one edge $\edge{v_i}{v_{i+1}}$,
see Fig.~\ref{fig:SC-graph}.
\item (*) An \emph{ambiguous separating tripod} is a strong separating tripod $C(v_1, v_2, v_3)$
 such that $G'$ decomposes into two parts, one of which  is
    a single vertex $u$ (of degree six), there are three bridges
    $\edge{u}{v_i} \in \mathcal{E}\edge{v_i}{v_{i+1}}$ for $i=1,2,3$
    with $v_4=v_1$ and there are edges $\edge{v_1}{v_2}$ and
    $\edge{v_2}{v_3}$, see Fig,~\ref{fig:strong-SC}.
  \end{enumerate}
\end{definition}

\begin{figure} [t]  
    \centering
  \subfigure[] {
     \includegraphics[scale=0.5]{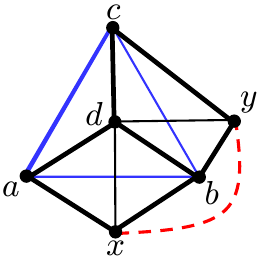}  
    \label{fig:ambigtriangle1}
  }
  \hfil
  \subfigure[] {
      \includegraphics[scale=0.5]{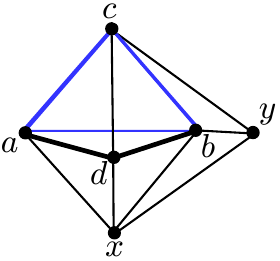} 
      \label{fig:ambigtriple2}
  }
 \subfigure[] {
     \includegraphics[scale=0.2]{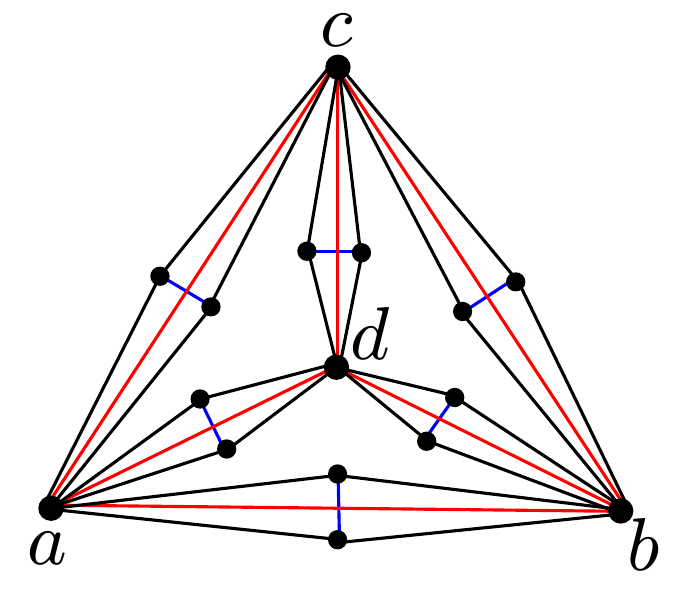}
    \label{fig:kite-covered-tetrahedron}
  }
\hspace{3mm}
  \hfil
  \subfigure[] {
      \includegraphics[scale=0.5]{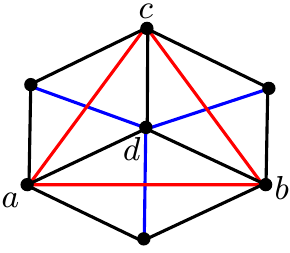}
      \label{fig:SC-graph}
  }
     \caption{(a) An ambiguous separating triangle $C(a,b,c)$ with
     center $d$  inside and (b) outside $C(a,b,c)$, and (c)
a completely kite-covered
   tetrahedron for a separating  tripod   and (d) a strong separating tripod with center $d$.
     }
     \label{fig:ambigtriangle}
\end{figure}

The order of generalized separators has the following effect, which
is useful   for the correctness proof.
  A separating edge  is searched only in
4-connected graphs, since 3-separators are searched before. There is no
separating edge if a separating triple is searched. However,
  unique separating quadrangles, bridge-independent and
singular separating triangles can be searched in any order, as opposed
to similar generalized separators in \cite{cgp-rh4mg-06},
and they precede their ambiguous counterparts.
At last
an ambiguous separating tripod can only be  found in the absence of
all other separators.
 We  claim that each  graph $G$, obtained by the use
 of generalized separators, is a
triangulated planar graph  if the given graph $G_0$ is a \emph{T1P}
graph.   Conversely, $G_0$ is a \emph{T1P} graph if each of its
components at separating $k$-cycles is reduced to a triangulated
planar graph. The \emph{T1P} embeddings of $G$ resp. $G_0$ are composed from the
\emph{T1P} embeddings of each component. We prove the claim
and thus Theorem~\ref{thm1}, by a
series of Lemmas, similar to the proof given by Chen et
al.~\cite{cgp-rh4mg-06}. They use maps for the presentation,
whereas we use 1-planar
embeddings, which is simpler, in particular for separating triangles.
In addition, we show that there
are ambiguities that are detected at generalized separators
and a reported by an edge label $\lambda(e)>1$.

\subsection{$k$-Separators} \label{sect:separators}
Forthcoming, we assume that $G$ is a  \emph{T1P} graph of order at
least seven. It is obtained from the given input graph by the use of
generalized separators, in general. The \emph{T1P} embeddings of small graphs
are first computed by exhaustive search, as suggested in \cite{cgp-rh4mg-06}.
We list them later and compute the number of
\emph{T1P} embeddings using generalized separator in Lemma~\ref{lem:small-graphs}.
There will be less \emph{T1P} embeddings if further   edges are marked, for example,
if a triangle is  the outer boundary of a separating 3-cycle, such that its  edges are marked.

We say that an edge is \emph{uncrossed} (\emph{crossed}) if this is
true in every \emph{T1P} embedding of $G$. Otherwise, it is crossed in one
\emph{T1P} embedding and uncrossed in another. Hence, all edges of $K_5$-$e$ are
uncrossed, since the 1-planar drawings in Figs.~\ref{fig:K5-e-aug}
and \ref{fig:K5-e-ill} are excluded.
We search for a generalized separator in the absence of some smaller separators,
which are listed and used before, such that a  generalized separator is \emph{minimal}
at the time of its use.

\begin{proposition}  \label{prop:all-separators}
 If $S$ is a set of $k$ vertices ($k =3,4,5$) in a $k$-connected T1P
graph $G$, such that $S$ is  a minimal generalized separator, then
$G[S]$ is an induced  separating $k$-cycle, that is the vertices of $S$
are ordered, any two consecutive vertices   are connected by an edge,
the $k$-cycle has no chord, its edges are uncrossed in every T1P
embedding, and
$G-S$ decomposes into two parts.
\end{proposition}

Proposition~\ref{prop:all-separators} has been proved (in parts) at
several places \cite{b-4mapGraphs-19, mrs-outermap-18,cgp-rh4mg-06}.
It is due to the fact that a \emph{T1P} graph (or a  hole-free 4-map
graph) has a planar bipartite graph as a witness \cite{cgp-mg-02}
and that the 1-planar embedding triangulated.
The exclusion of a  chord for $k=4,5$ is due to the absence of smaller
separators, in particular separating  triples
(Section~\ref{sect:triple}) for 4-cycles and   separating
quadruples (Section~\ref{sect:quadruple}) for 5-cycles. Note that
any edge  of a minimal cycle can be crossed in the given graph if
the crossing edge has been removed before.

In consequence, there is a clear decision at $k$-separators. Each
unmarked edge  $e$ of $G[S]$ is marked
and the pair $(e, \mathcal{E}(e))$ is removed
from the respective set computed in step (3).  If $G_1$ and $G_2$ are the parts
of $G-S$, then the components $G_1+S$ and $G_2+S$ are treated
recursively, such that $S$ is the outer face of an embedding of
$G_1+S$ and an inner face of $G_2+S$, or vice versa.
Each face is
triangulated by   marked chords   if $k \geq 4$. A chord   is
chosen such that it does not create a 5-clique in a component. This is
possible, since $G$ has a smaller separator if it has a separating
4-cycle with a chord. The chord is removed before the composition of
the embedding of the components to an embedding of $G$.  The number
of \emph{T1P} embeddings of $G$ is the product of the number of \emph{T1P}
embeddings of the components.
Hence, each vertex of degree three
can be treated in preprocessing step (2), since it is one part at a
separating 3-cycle, whose edges can be marked.

\subsection{Separating Edge} \label{sect:separtingedge}
Forthcoming, we consider 4-connected \emph{T1P} graphs of order at
least seven. This excludes  cliques of size six \cite{cgp-rh4mg-06},
since they have 3-cycles in the interior and the exterior, see~Fig.~\ref{fig:K6-triangle}.
  Next we show that a separating edge is crossed
in every \emph{T1P} embedding. It crosses one of its bridges, such
that there are  $t$ \emph{T1P} embeddings if there are $t$ bridges.
Conversely, if there are two \emph{T1P} embeddings  that coincide
except for the routing of  edge $\edge{u}{v}$, then $\edge{u}{v}$ is
a separating edge. These facts have been proved in
\cite{cgp-rh4mg-06} using maps.

\begin{lemma} \label{if-seperatingedge-cross}  
If $\edge{u}{v}$ is a separating edge, then it is crossed and it
crosses one of its $t \geq 2$ bridges.
\end{lemma}

\begin{proof}  
Graph $G'=G-\{u,v\}-\mathcal{E}\edge{u}{v}$ decomposes into
 parts with a set of bridges
$\mathcal{B}\edge{u}{v}= \{\edge{x_i}{y_i} \, | \, i=1,\ldots, t\}
\subseteq \mathcal{E}\edge{u}{v}$
if  $\edge{u}{v}$ is a   separating edge.
Vertices  $x_i$ and $y_i$   are in different parts and there
is a maximal  4-clique $K_i=K(u,v,x_i,y_i)$ for $1 \leq i \leq t$.
  Since
$G$ is 4-connected, it does not decompose by the removal of a 3-set
$\{u,v,x_i\}$ or $\{u,v,y_i\}$, but it decomposes if $u,v$ and the
edges of $\mathcal{B}\edge{u}{v}$ are removed.
Hence, there are at least two bridges.
We claim that there is no vertex $w \neq  u,v$
that is a neighbor of both $x_i$ and $y_i$.
Otherwise, at least one of $\edge{x_i}{w}$ or $\edge{y_i}{w}$
must be a crossable edge of $\edge{u}{v}$, since $x_i$ and $y_i$ are
connected via $w$ in $G'$, otherwise, see Fig.~\ref{fig:proof-separatingedge}.
 Then $K(u,v,x_i,y_i,w)$ is a 5-clique, since edges $\edge{w}{x_i}$
and $\edge{w}{y_i}$ must exists by crossability,
contradicting the maximality of $K(u,v,x_i,y_i)$.

However, a common neighbor $w \neq u,v$ of $x_i$ and $y_i$
cannot be avoided if $\edge{u}{v}$ is uncrossed  in an
embedding of $G$. Then each $K_i$ is represented as a tetrahedron or
$K_i$ is an X-quadrangle with edges $\edge{u}{v}$ and $\edge{x_i}{y_i}$
on the boundary.  In the latter case, there is a neighbor $w$ of $x_i$ and $y_i$,
since $G$ has at least five vertices and its embedding is triangulated.
If $K_i$ is represented
as a tetrahedron, 
 then  there are no vertices on both sides of  $C(u, x_i, y_i)$ and $C(v,x_i,y_i)$,
such that  the edges of $C(u, x_i, y_i)$ and $C(v,x_i,y_i)$ are uncrossed.
Otherwise,
 one of  the edges of  $C(u, x_i, y_i)$ must be crossed by 4-connectivity. As before,
if $\edge{x_i}{y_i}$ is crossed by an edge $\edge{w}{w\rq{}}$, then one of
$\edge{x_i}{w}$ or $\edge{y_i}{w}$ is in $\mathcal{E}\edge{u}{v}$,
such that $K(u,v,x_i,y_i,w)$ is a 5-clique. Otherwise, $x_i$ and $y_i$ have
a common neighbor $w$ in the interior of $C(u, x_i, y_i)$,
contradicting our claim.
The case of $C(v,x_i,y_i)$
is similar. Now, one of $C(u,v,x_i)$  or $C(u,v,y_i)$ is a separating 3-cycle
if  $\edge{u}{v}$ is uncrossed that separates $x_i$ or $y_i$, respectively,
from a fifth vertex of $G$, a contradiction.

Hence, edge $\edge{u}{v}$ is crossed and there is exactly one $i$, $1\leq i
\leq t$, such that $K_i= Q(u,x_i,v,y_i)$ is an X-quadrangle.
\qed
\end{proof}

\begin{figure}[t]
  \centering
     \includegraphics[scale=0.7]{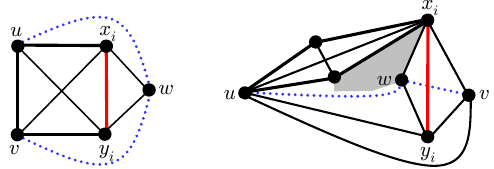}  
   \caption{Illustration for the proof of Lemma~\ref{if-seperatingedge-cross}
with a bridge $\edge{x_i}{y_i}$ and a (nearest) neighbor $w$.
}
  \label{fig:proof-separatingedge}  
\end{figure}

Lemma~\ref{if-seperatingedge-cross} shows that a separating edge is
as displayed in
Fig.~\ref{fig:separatingedge1}. In particular, $K_5$ has a separating edge.
Moreover, it suffices to consider graphs of order at least seven, whereas
Chen et al.~\cite{cgp-rh4mg-06} use graphs of order at least nine.

\begin{lemma} \label{then-seperatingedge}  
 Suppose there are two \emph{T1P} embeddings $\Gamma_1(G)$ and
$\Gamma_2(G)$, such that  edge $\edge{u}{v}$ crosses   $\edge{x}{y}$
in $\Gamma_1(G)$ and   $\edge{x'}{y'}$ in $\Gamma_2(G)$. Otherwise,
$\Gamma_1(G)$ and  $\Gamma_2(G)$ coincide, that is on
$G-\edge{u}{v}$. Then $\edge{u}{v}$  is a separating edge.
\end{lemma}

\begin{proof}  
We merge $\Gamma_1(G)$ and $\Gamma_2(G)$ and obtain a 1-planar
embedding   $\Gamma'(G)$ with a multi-edge $\edge{u}{v}$. Edge
$\edge{u}{v}$ and its copy are a closed curve $J$ in $\Gamma'(G)$
that encloses a subgraph $G_1$. Curve $J$ is crossed twice, namely  by
$\edge{x}{y}$  and $\edge{x'}{y'}$, such that $x$ and $y$ are on
opposite sides of $J$, and similarly for $x'$ and $y'$. Note that $x
\neq x'$ if $y=y'$, or vice versa.
The removal of vertices $u$ and $v$ and edges $\edge{x}{y}$ and
$\edge{x'}{y'}$ partitions $G$ into parts $G_1$ and $G_2$ such that
$G_1$ contains $x, y$ and $G_2$ contains $x', y'$. Since our
1-planar embeddings are triangulated, there are 4-cliques
$K=K(u,v,x,y)$ and $K'=K(u,v,x',y')$. If $K$ and $K'$ are maximal
4-cliques, then edges $\edge{x}{y}$ and $\edge{x'}{y'}$ are
crossable edges of $\edge{u}{v}$.
 Towards a contradiction, assume that $K$ is not
maximal. If  $K(u,v,w,x,y)$ is a 5-clique, then consider the
placement of $w$ in the embedding $\Gamma'(G[u,v,x,y,x',y',w])$.
Since $J$ is a closed curve formed by $\edge{u}{v}$ and its copy, there is a
closed curve $C(u,x,v,x')$ of uncrossed edges, and similarly for
$C(u,y,v,y')$. Then
 there is no face in $\Gamma'(G[u,v,x,y,x',y',w])$, such that
 $w$ can be placed in the face and
1-planarity is preserved.

Hence, $\edge{u}{v}$ is a separating edge with bridges $\edge{x}{y}$  and $\edge{x'}{y'}$.
\qed
\end{proof}

\begin{lemma} \label{lem:separatingedge-reduce}
Let $\edge{u}{v}$ be  a separating edge with a set of bridges
$\mathcal{B}\edge{u}{v}= \{\edge{x_i}{y_i} \, |\, i=1,\ldots, t\}$ for some $t
\geq 2$.
If $\Gamma_i(G)$ is a T1P embedding such that $\edge{u}{v}$ crosses
$\edge{x_i}{y_i}$ for some $1 \leq i \leq t$, then there is an
X-quadrangle $Q(u,x_i,v,y_i)$  and there are triangles $C(u,
x_j,y_j)$ and $C(v, x_j, y_j)$ for $j\neq i$ and $1 \leq j \leq t$.
There is a  T1P  embedding of $G- \edge{u}{v}$ with  triangles $C(u,
x_j,y_j)$ and $C(v, x_j, y_j)$ for $j=1,\ldots, t$.
Conversely, if $\Gamma'(G- \edge{u}{v})$ is a  T1P  embedding such
that there are triangles $C(u, x_j,y_j)$ and $C(v, x_j, y_j)$ for
$j=1,\ldots, t$, then there are   T1P  embeddings $\Gamma_1(G),
\ldots, \Gamma_t(G)$ such that  $\edge{u}{v}$ crosses
$\edge{x_i}{y_i}$ in $\Gamma_i(G)$ and $\edge{x_j}{y_j} \in
\mathcal{B}\edge{u}{v}$ with $i\neq j$ is uncrossed in $\Gamma_i(G)$.
\end{lemma}

\begin{proof}
By Lemma~\ref{if-seperatingedge-cross}, edge $\edge{u}{v}$ is
crossed in every embedding of $G$, and it crosses a bridge
$\edge{x_i}{y_i}$, such that there is an X-quadrangle
$Q(u,x_i,v,y_i)$. If $\edge{u}{v}$ is removed, then there are triangles
$C(u,x_i,y_i)$ and $C(v,x_i,y_i)$ for $i=1,\ldots, t$.

Conversely, edge $\edge{u}{v}$ can be inserted into $\Gamma'(G-
\edge{u}{v})$ such that it crosses $\edge{x_j}{y_j}$ for any $j$
with $1\leq j\leq t$. Then there is an X-quadrangle
$Q(u,x_j,v,y_j)$. Hence, there is a \emph{T1P} embedding
$\Gamma_j(G)$ for $j=1,\ldots,t$.
\qed
\end{proof}

 There is a clear decision for algorithm $\mathcal{N}$ if there
is a separating edge $\edge{u}{v}$ with $t \geq 2$ bridges
$\edge{x_i}{y_i}$. We set $\lambda(\edge{u}{v})=t$ and remove
$\edge{u}{v}$. For each bridge $\edge{x_i}{y_i}$ we mark
$\edge{x_i}{y_i}$ and edges $\edge{u}{x_i}, \edge{v}{x_i},
\edge{u}{y_i}$ and $\edge{u}{y_i}$ if they are unmarked,
and   the pair $(e, \mathcal{E}(e))$ is removed for each newly marked edge $e$.
In the next steps, algorithm $\mathcal{N}$
partitions $G-\edge{u}{v}$ into $t$ components using separating
4-cycles $C_i(u,x_i,v, y_{i-1})$ with $y_0=y_t$ for $i=1,\ldots,t$.

\subsection{Separating Triples} \label{sect:triple}

Forthcoming, let $G$ be a \emph{T1P} graph of order at least seven,
such that $G$ is 4-connected and does not have a separating edge.
Intuitively, the following is clear. If $G$ does not decompose by a
 3-cycle $C(a,b,c)$, but by a 3-cycle and a crossable edge $e$ of
 $\edge{a}{b}$ 
Then $e$  is a bridge that connects the  two parts in the interior and the
exterior of $C(a,b,c)$, such that the 4-clique $K(a,x,b,y)$ is represented
as an X-quadrangle in which $\edge{x}{y}$ crosses $\edge{a}{b}$,
as shown in Fig.~\ref{fig:triple}.
This fact has been shown by
 Chen et al.~\cite{cgp-rh4mg-06} (Lemma 7.7). We prove it using \emph{T1P}
embeddings and call it the ``Single Bridge Lemma''.
For the proof, we use the following technique.
Suppose vertex $w$ is a new neighbor of $x$ and $y$,
that is $w \not\in \{a,b,c,x,y\}$. Then at least one of $\edge{w}{x}$ and $\edge{w}{y}$
is a crossable edge of $\edge{a}{b}$. 
 Otherwise, $x$ and $y$ are connected via $w$, such that edge $\edge{x}{y}$ is not a bridge. If $\edge{w}{x}$ is a crossable edge of $\edge{a}{b}$,
then $\edge{a}{b}$ is a separating edge isolating vertex $x$, which is
excluded by assumption, or there is a 5-clique containing $\{a,b,x,y\}$ or $\{a,b,x,w\}$ if
$\edge{w}{x}$ is crossable, which contradicts the maximality of 4-cliques for
crossable edges.
There is a Single Bridge Lemma also
  for unique separating quadrangles, bridge-independent and singular separating triangles, and separating tripods.

Note that $\edge{a}{b}$ is any edge of the 3-cycle with vertices $a,b$ and $c$.

\begin{lemma} \label{lem:single-bridge-lemma}    
If $C(a,b,c)$ is a separating triple, then there is a single bridge
$\edge{x}{y} \in \mathcal{E}\edge{a}{b}$, that is
$G''=G-C(a,b,c)-\edge{x}{y}$ is disconnected if so is
$G'=G-C(a,b,c)-\mathcal{E}\edge{a}{b}$, and there is no edge $e' \neq \edge{x}{y}$
such that $G-C(a,b,c)-e'$ is disconnected. In addition, vertices $x$ and
$y$ are on opposite sides of $C(a,b,c)$, edges $\edge{a}{b}$ and $\edge{x}{y}$ cross,
whereas $\edge{b}{c}$ and $\edge{a}{c}$ are uncrossed,
and $G'$ and $G''$ has two parts that are connected by $\edge{x}{y}$.
\end{lemma}

\begin{proof}
Since $G$ is 4-connected, the removal of vertices $a,b,c$ does not
decompose $G$, whereas $G'$ has
 at least two parts.
Let $\edge{x}{y} \in \mathcal{E}\edge{a}{b}$ be a bridge with
vertices $x$ and $y$   in different parts of $G'$. Then there is a
   4-clique $K(a,b,x,y)$.
We claim that $K(a,b,x,y)$ cannot be represented as a tetrahedron in
$\Gamma(G)$ nor as an X-quadrangle $Q(a,x,y,b)$.
In the first case, if $x$ and $y$ are on opposite sides of $C(a,b,c)$,
 then
  edge $\edge{x}{y}$ crosses $\edge{a}{c}$ or $\edge{b}{c}$,
such that $G[a,b,c,x,y]$ is a 5-clique,
contradicting the assumption that $K(a,b,x,y)$ is a maximal 4-clique.
 Thus suppose that $x$ and $y$ are
outside $C(a,b,c)$ and let $x$ be the center of a tetrahedron for
$K(a,b,x,y)$. The case with any other center    is similar.

As before, $x$ and $y$ do not have a common neighbor $w$, since
that implies  a 5-clique containing $K(a,b,x,y)$. Here, $w=c$ is possible.
Hence,  there are no vertices
both in the interior and exterior of the triangles $C(a,x,y)$ and
$C(b,x,y)$, such that the edges in the boundary of these
triangles are uncrossed. Now there is a separating 3-cycle $C(a,b,z)$, where
$z$ is one of $x,y$ or $c$, if edge $\edge{a}{b}$ is uncrossed,
or $\edge{a}{b}$  is a separating edge if it is crossed, as it can also cross
$\edge{x}{y}$.
 In any case, there is a contradiction to the assumption.
The case for an X-quadrangle $Q(a,x,y,b)$ is similar.

Suppose there is another bridge $ \edge{x\rq{}}{y\rq{}}\in
 \mathcal{E}\edge{a}{b}$.  By the reasoning as before,
it crosses in an X-quadrangle $Q(a,x\rq{}, y\rq{},b)$,
which implies $\edge{x}{y}=\edge{x\rq{}}{y\rq{}}$ by
1-planarity. Hence, there is a unique bridge.
Clearly, vertex $c$ cannot be in the interior (exterior) of $C(a,x,b,y)$
if edges $\edge{a}{b}$ and $\edge{x}{y}$ cross in the interior (exterior).
Hence, vertices $x$ and $y$ are on
opposite sides of $C(a,b,c)$. Next, if $\edge{a}{c}$ is crossed
by  an edge $e$, then vertices $x$ and $y$ are connected via $e$
in $G'$ and $G''$, respectively, a contradiction.
Hence, edges $\edge{a}{b}$ and $\edge{a}{c}$ are uncrossed.
Then, $G'$ and $G''$, respectively, have exactly two parts.
\qed
\end{proof}

\begin{lemma} \label{lem:embed-triple}   
 If $C(a,b,c)$ is a separating triple such that  $\edge{a}{b}$ has the
bridge   $\edge{x}{y}$,
 then $\Gamma(G)$ is a T1P
embedding if and only if $\Gamma(G-\edge{x}{y})$ is a T1P embedding
with triangles $C(a,b,x)$ and $C(a,b,y)$.
\end{lemma}
\begin{proof}
There is an X-quadrangle $Q(a,x,b,y)$ by Lemma
\ref{lem:single-bridge-lemma}. Then there are triangles $C(a,x,b)$
and $C(a,y,b)$ if $\edge{x}{y}$ is removed from  $\Gamma(G)$. Thus
$\Gamma(G-\edge{x}{y})$ is a \emph{T1P} embedding.
Conversely, if a \emph{T1P} embedding $\Gamma(G-\edge{x}{y})$ has triangles
$C(a,b,x)$ and $C(a,b,y)$, then edge $\edge{a}{b}$ is uncrossed in
$\Gamma(G-\edge{x}{y})$. Now edge $\edge{x}{y}$ can be inserted
into the embedding such that it
crosses $\edge{a}{b}$, which is a \emph{T1P} embedding of
$G$.
\qed
\end{proof}

Hence, there is a clear decision for a separating triple. If
$\edge{x}{y}$ is the bridge of the separating triple $C(a,b,c)$,
which is found in $G'=G-C(a,b,c)-\mathcal{E}\edge{a}{b}$,
then algorithm $\mathcal{N}$ marks the edges of $C(a,b,c)$ and
 $\edge{a}{x}, \edge{a}{y}, \edge{b}{x}$ and $\edge{b}{y}$ if
they are unmarked before. Then edge $\edge{x}{y}$ is removed and
$\mathcal{N}$ proceeds on $G-\edge{x}{y}$, where it discovers the separating
3-cycle $C(a,b,c)$ at the next step.\\

Chen et al.\cite{cgp-rh4mg-06} exclude 5-cliques from generalized
separators and study their representation only in 5-connected
graphs. Then 5-cliques are either treated as part of a small graph
or there are three $K_5$ that mutually share at least three
vertices. In   case of small graphs, an embedding is computed by
exhaustive search, such that the number of embeddings is not
immediately clear. We proceed in a different way and show that
5-cliques are classified by the degree of the inner vertices.

\begin{figure}[t]   
  \centering
  \subfigure[] {
     \includegraphics[scale=0.40]{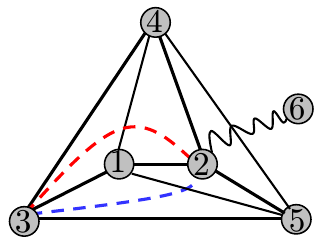}  
    \label{fig:divedge1}
  }
  \hfil
  \subfigure[] {
      \includegraphics[scale=0.4]{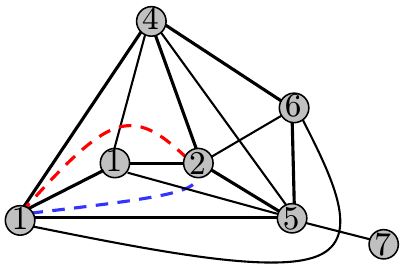}
      \label{fig:divedge2}
  }
  \caption{Illustration for the proof of Lemma~\ref{lem:create-h2S}
  with a vertex of degree four in a 5-clique.
  (a)   Edge $\edge{v_2}{v_3}$ can be routed such that it crosses
  $\edge{v_1}{v_4}$ or  $\edge{v_1}{v_5}$. (b) Edge $\edge{v_2}{v_6}$
  of the 5-clique is crossed by 4-connectivity.
  }
  \label{fig:h2S-create}
\end{figure}

\begin{lemma} \label{lem:create-h2S}
If $G$ has a $K_5$-triple and no smaller separators,
then it is a generalized handle two-star, that is $G=hG2S_k$.
\end{lemma}

\begin{proof}
Let $G[U]$ with $U= \{u_1,\ldots, u_5\}$ be a 5-clique  such that
$u_1$ has degree four and there is no smaller
generalized separator, in particular, no separating triple. Vertex $u_1$ is isolated from the remainder
if three vertices $u_2, u_3$ and $u_4$ and edge $\edge{u_1}{u_5}$
are removed. In other words, there are  3-cycles and an edge in a
5-clique whose removal decomposes $G$.  Vertex $u_1$ is the corner
of two triangles and an X-quadrangle in $\Gamma(G)$, since a vertex
$v \not\in U$ in a triangle with  $u_1$ in its boundary would
increase its degree
  by the triangulation. Let $u_6 \not\in U$ be another vertex of $G$, which
 is in a triangle of $\Gamma(G[U])$ that is chosen as the outer face of
$\Gamma(G[U])$, which is as shown in Fig.~\ref{fig:K5variation}. Then $u_1$ is an inner vertex of $\Gamma(G[U])$. Let $u_2$
be the other inner vertex such that $C(u_3,u_4, u_5)$ is the outer
boundary of $\Gamma(G[U])$.

By 4-connectivity, there are four vertex disjoint paths between
$u_1$ and $u_6$. Hence, there is a path $\pi_{i-1}$ from $u_1$ to
$u_6$ that passes through $u_i$ and not through $u_j$ for $j \neq i$
and $i=2,\ldots,5$. Now $\pi_2$   crosses an outer side of
$\Gamma(G[U])$, for example $\edge{u_4}{u_5}$. Let $\edge{x}{y}$ be the
 edge crossing $\edge{u_4}{u_5}$. Since separating triples are excluded,
$x$ and $y$ are in a 5-clique
$G[u_4,u_5,x,y,z]$ for some vertex $z$. Since there is no $K_6$, we
have $y=u_6$ and $x=u_2$. Thus there are 5-cliques $G[u_1,\ldots,
u_5]$ and $G[u_2,\ldots,u_6]$, such that   $\edge{u_4}{u_5}$ crosses
$\edge{u_2}{u_6}$, see Fig.~\ref{fig:h2S-create}. Consider another vertex $u_7$ of $G$. Since it is
not in a face with $u_1$ in its boundary and edge $\edge{u_2}{u_6}$
is crossed, $u_7$ is outside  $C(u_3,u_4,u_6,u_5)$. If $G$ has order seven, then
there are edges $\edge{u_7}{u_i}$ for $i=3,4,5,6$ by 4-connectivity,
such that $G=hG2S_5$. Otherwise, we proceed by induction. By the
absence of separating triples and smaller separators, vertex
$u_{i+1}$ for $i \geq 7$ is in a 5-clique with vertices $p,q,
u_{i-1}, u_i, u_{i+1}$, where $p=u_3$ and $q =u_4$ or $q=u_5$.
Vertices $p$ and $q$ are distinguished by their degree, which is at
least seven, whereas all other vertices have degree at most six.
Then edge $\edge{u_{i-1}}{u_{i+1}}$ crosses either $\edge{p}{u_i}$
or $\edge{q}{u_i}$.   Thus $G=hG2S_k$ with vertices $p$ and $q$ and
vertex $u_k$ of degree four in the outer boundary of $\Gamma(G)$.
\qed
\end{proof}

\begin{figure}[t]   
  \centering
  \subfigure[] {
     \includegraphics[scale=0.5]{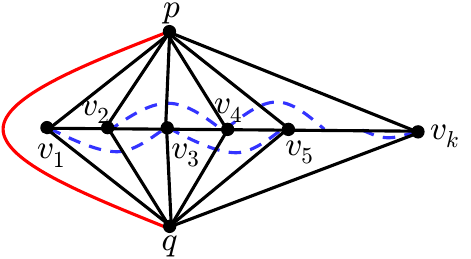}  
    \label{fig:hG2Sa}
  }
  \hfil
  \subfigure[] {
      \includegraphics[scale=0.5]{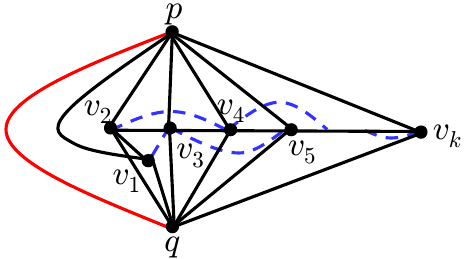}    
      \label{fig:hG2Sb}
  }
  \caption{Illustration for the proof of Lemma~\ref{lem:embed-h2S}, showing two embeddings
for $hG2S_k$. Two more embeddings are obtained if  $p$ and $q$ swap.
  }
  \label{fig:h2S-embeddings}
\end{figure}

\begin{lemma} \label{lem:embed-h2S}    
A  generalized handle two-star $hG2S_k$ with poles $p$ and $q$ and
vertices $v_1, \ldots, v_k$ for $k\geq 5$  has four, two or one  T1P
embeddings or there is an error if every marked edge is uncrossed. There is a single embedding if the
edges from one of the 3-cycles $C(p,v_1, v_2)$, $C(q, v_1, v_2)$,
$C(p,v_1, v_3)$ or $C(q, v_1, v_3)$ are marked, there are two embeddings if the edges of $C(p, q,v_1)$
are marked, and there are four embeddings, if none of these edges is
marked. There is an error, otherwise, for example, if edge
$\edge{p}{v_i}$ or $\edge{q}{v_i}$ for $i=4,\ldots, n-1$ is marked.
\end{lemma}

\begin{proof}
If $k \geq 6$, then the degree of the vertices determines the poles
$p$ and $q$, which have degree at least seven. There are two
vertices of degree four, namely $v_1$ and $v_k$, vertices $v_2$ and
$v_{k-1}$ have degree five and are neighbors of $v_1$ and $v_k$,
respectively. The remaining vertices $v_j$ for $j=3,\ldots, k-2$
have degree six. Their order is determined inductively from $v_1$ to
$v_k$. Otherwise, check for an $hG2S_k$ by inspection.

First, suppose there are no restrictions for uncrossed edges, that is there are
no marked edges. There
are three 5-cliques $G[p,q,v_1,v_2,v_3], G[p,q,v_2,v_3,v_4]$
and $G[p,q,v_3,v_4,v_5]$. Then edges $\edge{v_2}{v_4}$ and
$\edge{v_3}{v_5}$ are crossed such that $\edge{v_2}{v_4}$ crosses
$\edge{p}{v_3}$ if  $\edge{v_3}{v_5}$ crosses $\edge{q}{v_4}$, and
vice versa.
 If $\edge{v_2}{v_4}$ crosses $\edge{p}{v_3}$, then
 vertex $v_1$ can be placed inside $C(p,q,v_2)$, as shown in Fig.~\ref{fig:hG2Sa}, such that
 $\edge{v_1}{v_3}$ is crossed and it crosses $\edge{q}{v_2}$,
or $v_1$ is inside $C(q,v_2, v_3)$ such that
$\edge{p}{v_1}$ and $\edge{q}{v_2}$ cross,   as shown in Fig.~\ref{fig:hG2Sb}.
We now proceed by induction and incrementally  add vertex $v_{i+1}$
from the 5-clique $G[p,q,v_{i-1}, v_i, v_{i+1}]$ in the outer face
of $C(p,q, v_i)$, such that $\edge{v_{i-1}}{v_{i+1}}$ crosses
$\edge{p}{v_i}$ if $\edge{v_{i-2}}{v_{i}}$ crosses
$\edge{q}{v_{i-1}}$, and $\edge{v_{i-1}}{v_{i+1}}$ crosses
$\edge{q}{v_i}$ if $\edge{v_{i-2}}{v_{i}}$ crosses
$\edge{p}{v_{i-1}}$ for $i=3,\ldots, n-1$. Thereby, we obtain four
\emph{T1P} embeddings. The first three 5-cliques do not allow
another
  \emph{T1P} embedding. By induction, there is a unique
extension of a \emph{T1P} embedding of $G[p,q, v_i,\dots, v_i]$ to a
\emph{T1P} embedding of $G[p,q, v_i,\dots, v_{i+1}]$.

 If
one of the edges $\edge{p}{v_i}, \edge{q}{v_i}, i=1,\ldots, k-1$ or
$\edge{v_1}{v_3}$ is marked, such that it is  uncrossed in the
embedding, then just two \emph{T1P} embeddings remain.
There is a unique \emph{T1P} embedding if $\edge{v_1}{v_3}$  and one
of $\edge{p}{v_i}, \edge{q}{v_i}, i=1,\ldots, k-1$ are uncrossed.
Otherwise, there is an error, for example, if $\edge{p}{v_i}$ and
$\edge{q}{v_i}$ are marked and shall be uncrossed. Note that, in total, there are
six triangles in the possible \emph{T1P} embeddings of $h2S_k$, and
that the edges of a triangle are marked if there is a separating
3-cycle (at an earlier step).
\qed
\end{proof}

In general, a generalized handle two-star $G2S_k$ occurs as a component after using
a separating 3-cycle.
If there is a
$K_5$-triple, then
algorithm $\mathcal{N}$ may immediately check  whether or not
  $G=hG2S_k$ by using the degree of the vertices if $n \geq 7$.
If so, it removes all edges
$\edge{v_{i-1}}{v_{i+1}}$
  towards planarity  and sets
  $\lambda(e)=t$ for a single edge of $h2S_k$ if it has $t$
embeddings, where $t=1, 2$ or $t=4$. The other variations of generalized handle two-stars
occur after using a particular generalized separator that   removes $\edge{v_1}{v_k}$
or $\edge{v_1}{v_{k-1}}$.

\subsection{Separating 4-Cycles} \label{sect:4cycle}

Next we search for 4-separators. By
Proposition~\ref{prop:all-separators}, a 4-separator with four
vertices induces a separating 4-cycle $C$ in a \emph{T1P} graph  if
there are no smaller separators, in particular, if there is no separating
triple. The edges of $C$ are marked, which means that they are
uncrossed,  and the set with pairs
$(e, \mathcal{E}(e))$ is updated. For each component, a marked chord
is added for the triangulation. It can be chosen
such that it does not help to create a 5-clique,   since
a 5-clique $K(a,b,c, x,y)$ excludes a 5-clique with vertices $a,c,d$
by 1-planarity.

\subsection{Separating  Quadruples} \label{sect:quadruple}

From now on, we assume that graphs are 5-connected \emph{T1P} graphs   of order at
least seven. Then there are no smaller separators, that is no
 separating edges,  separating triples  or separating 4-cycles.

As an extension from separating triples, the following is obvious:
If $C(a,b,c,d)$ is a 4-cycle such that
edge $\edge{a}{b}$ is crossed by an edge $\edge{x}{y}$ and the other
edges $\edge{a}{d}, \edge{b}{c}$ and $\edge{c}{d}$ are uncrossed
in a \emph{T1P} embedding $\Gamma(G)$, and there is no 5-clique containing
vertices $a,b,x,y$,
then $C(a,b,c,d)$  is a separating quadruple.
Chen et al.~\cite{cgp-rh4mg-06} state this as Fact 7.8.
For the converse, they write:
 "So we can modify the proof of Lemma 7.7
(on separating triples)  to prove the following: (Lemma 7.9)
  \emph{If $C(a,b,c,d)$ is a separating quadruple, that is, $G'=
G-C(a,b,c,d)-\mathcal{E}\edge{a}{b}$ is disconnected, then  $G'$ has exactly
two connected parts,  exactly one bridge $\edge{x}{y}$ connects the parts,
and edges   $\edge{a}{b}$
and  $\edge{x}{y}$ cross, such that there is an X-quadrangle $Q(a,x,b,y)$."}

However, that's wrong!
Fig.~\ref{fig:quad-D} shows that three bridges are
possible and Figs.~\ref{fig:quad-many-bridges}, \ref{fig:quad-two-bridges}
and \ref{fig:quadrangles}
illustrate that an edge is not necessarily crossed by its bridge.
Edge $\edge{x}{y}$ is the bridge of any edge in $C(a,b,c,d)$ for the
graph in Fig.~\ref{fig:quad-C}, but it is uncrossed in every 1-planar
embedding. The problems are caused by the shown graphs, as we shall prove.
Second, there are ambiguous separating quadruples, as shown in
Fig.~\ref{fig:quad-ambig}, in which   bridge $\edge{x}{y}$
can choose to cross one of $\edge{a}{b}$ or $\edge{b}{c}$. There is no chord $e=\edge{a}{c}$, such that $G[a,b,c,x,y]$ is $K_5$-$e$.
Then the rules by Chen et al. can be applied
such that their algorithm $\mathcal{A}$ reports a failure for
generalized circle two-star graphs $cG2S_k$   although
the  graphs are \emph{T1P} graphs. All in all,  algorithm $\mathcal{A}$
is false positive.\\

  We consider separating quadruples in more detail, specialize them, and  distinguish
unique and ambiguous ones.
Let $C=C(v_1,v_2,v_3,v_4)$ be a 4-cycle with edges $e_i = \edge{v_i}{v_{i+1}}$ with
$v_5=v_1$ for $i=1,2,3,4$, such that $v_1$ is the first vertex     and $e_i$ is the $i$-th edge.
Let $G'_i=G-C-\mathcal{E}(e_i)$, that is $G'_i$ is the decomposition
of $G$ by a separating quadruple with the crossable edges of edge $e_i$.
Let $\ell_i$ denote the number of bridges of $G'_i$. For $j\neq i$ let $\ell_j$ denote the
number of bridges of $G'_i$ that are also a bridge of $e_j$, that is, there is
a 4-clique $K(v_j, v_{j+1}, x,y)$ for vertices $x,y$ not in $C$.

For $i=1,2,3,4$, the $i$-th \emph{bridge pattern} $\beta_i(C)$ is a
4-tuple $(\ell_1, \ell_2, \ell_3, \ell_4)$, such that $\ell_j$ is the number
of bridges of $e_j$ in $G'_i$ if $G'_i$ is disconnected. Otherwise,
if $G'_i$ is connected, set $\beta_i(C)=($-,-,-,-$)$.
Then $C$ is a separating 4-cycle if and only if  $\beta_i(C)=(0,0,0,0)$. It is
a unique separating quadruple if    $\beta_1(C)=(1,0,0,0)$ and
$\beta_j(C)=($-,-,-,-$)$ for $j=2,3,4$,
and an ambiguous separating quadruple if    $\beta_1(C)=\beta_2(C)=(1,1,0,0)$
and $\beta_3(C)=\beta_4(C)=($-,-,-,-$)$.
Even three bridges are possible, such that $G'_1 = G-C-\mathcal{E}(e_1)$ has parts
$x$, $y$, and a remainder (outside $C$), as shown in Figs.~\ref{fig:quad-D}
and \ref{fig:quad-2-bridges}.

\begin{figure}[t]    
  \centering
\subfigure[$H_7$] {
     \includegraphics[scale=0.5]{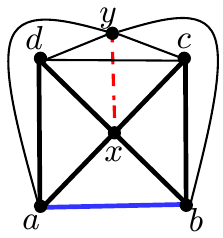}  
    \label{fig:quad-A}
  }
  \hspace{2mm}
  \subfigure[] {
     \includegraphics[scale=0.45]{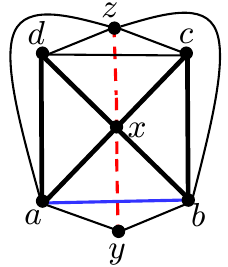}  
    \label{fig:quad-B}
  }
  \subfigure[] {
      \includegraphics[scale=0.45]{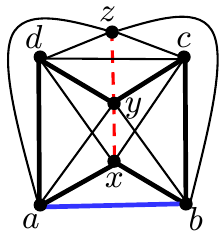} 
      \label{fig:quad-C}
  }
  \subfigure[] {
      \includegraphics[scale=0.45]{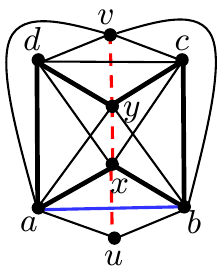} 
      \label{fig:quad-D}
  }
  \subfigure[] {
      \includegraphics[scale=0.45]{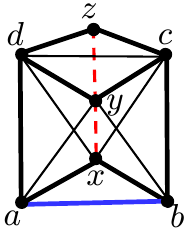} 
      \label{fig:quad-E}
  }
  \subfigure[] {
      \includegraphics[scale=0.45]{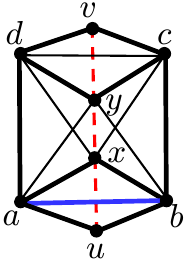} 
      \label{fig:quad-F}
  }
  \caption{
  A separating quadruple $C=C(a,b,c,d)$ with first edge $\edge{a}{b}$, where
 \newline (a) $\edge{x}{y}$ (drawn dashed and red) is a   bridge of  each edge of $C$,
and it is the only bridge, such that
  $\beta_i(C)=(1,1,1,1)$ for $i=1,2,3,4$.
\newline (b)
The first edge $\edge{a}{b}$ has two bridges and no other edge
has a bridge, such that $\beta_1(C)=(2,0,0,0)$ and $\beta_i(C)=($-,-,-,-$)$, otherwise.
\newline (c)
Edges $\edge{x}{y}$ and $\edge{y}{z}$ are bridges of $\edge{a}{b}$ and $\edge{c}{d}$
and $\edge{x}{y}$ is a bridge of $\edge{b}{c}$ and $\edge{a}{d}$, such that
$\beta_1(C)=\beta_3(C)=(2,0,2,0)$ and   $\beta_2(C)=\beta_2(C)=(1,1,1,1)$.
Graph $G'_1$ has three parts, namely vertices $x$ and $y$ and a remainder
outside $C(a,b,c,d)$,
\newline (d)
Edge $\edge{a}{b}$ has three bridges, and  $\edge{x}{y}, \edge{y}{v}$ are bridges of
$\edge{c}{d}$, such that $\beta_1(C)=(3,0,2,0), \beta_3(C)=(2,0,2,0)$ and
$\beta_2(C)=\beta_4(C)=($-,-,-,-$)$.
\newline (e) Edge $\edge{x}{y}$ is a bridge of  each edge of $C$
  and $\edge{c}{d}$ has two bridges, such that
$\beta_3(C)=(1,1,2,1), \beta_j(C)=(1,1,1,1)$ for $j=1,2,4$.
(f) Edge  $\edge{x}{y}$ is crossable for  each edge of $C$, but only
$\edge{a}{b}$ and $\edge{c}{d}$ have bridges, such that  $\beta_1(C)=(2,0,1,0)$
and $\beta_3(C)=(1,0,2,0)$.
}
  \label{fig:quad-many-bridges}
\end{figure}

\begin{figure}[t]    
  \centering
  \subfigure[] {
 \includegraphics[scale=0.5]{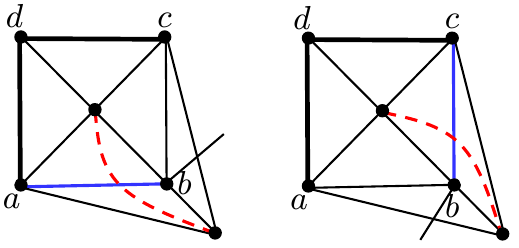} 
      \label{fig:quad-ambig}
  }
  \hfil
  \subfigure[] {
 \includegraphics[scale=0.5]{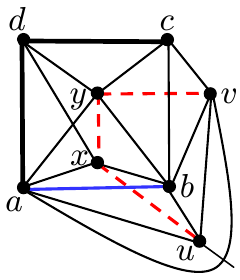}  
    \label{fig:quad-2-bridges}
  }
  \caption{  (a) An ambiguous   separating quadruple $C(a,b,c,d)$
  with a single bridge for $\edge{a}{b}$ or $\edge{b}{c}$.
There is no $K_5$, since the chords $\edge{a}{c}$ and $\edge{b}{d}$
are missing.
\newline (b)
 edge $\edge{a}{b}$ has
 three bridges, drawn dashed and red, which are each a crossable edge of another edge of $C$, such that
$\beta_1(C)=(3,2,0,0)$.  However, $\beta_i(C)=($-,-,-,-$)$, for $i=2,3,4$.
  }
  \label{fig:quad-two-bridges}
\end{figure}

\begin{figure}[t]  
  \centering
     \includegraphics[scale=0.5]{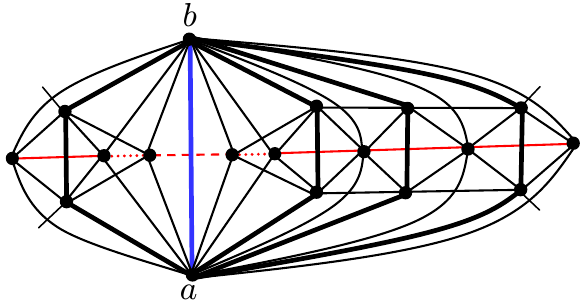}  
  \caption{Separating quadruples  with edge $\edge{a}{b}$ drawn
  bold and blue and crossable edges in $\mathcal{E}\edge{a}{b}$ drawn red.
  If the dashed edge is missing, then $\edge{a}{b}$ is not crossed
  by any of its bridges. Each bridge of $\edge{a}{b}$, drawn red and solid,
is also a crossable edge  of the opposite side of a 4-cycle.
Uncrossed crossable edges of $\edge{a}{b}$ are drawn dotted and red.
  }
  \label{fig:quadrangles}
\end{figure}

\begin{lemma} \label{lem:single-bridge-lemma-quad}    
 If $C=C(a,b,c,d)$ is a  unique separating quadruple,  that is $\beta_1(C)=(1,0,0,0)$, then
edge $\edge{a}{b}$  and its bridge $\edge{x}{y}$ cross, such
 that is there is an X-quadrangle $Q(a,x,b,y)$, whereas the other edges of $C(a,b,c,d)$ are uncrossed.
\end{lemma}

\begin{proof}
 Since $G$ is 5-connected, the removal of $C=C(a,b,c,d)$ does not decompose $G$,
whereas $G''=G-C - \edge{x}{y}$ is disconnected.
Then $G''$ has two parts that are connected by $\edge{x}{y}$.
If there is a 5-clique $K(a,b,x,y,z)$ for some vertex $z$, then $x$ and $y$ are
connected via $z$ in $G''$, so that $\edge{x}{y}$ is not a bridge.
Hence, $K(a,b,x,y)$ is a 4-clique, so that
$\edge{x}{y}$ is a crossable edge of $\edge{a}{b}$ and
 $G'=G-C -\mathcal{E}\edge{a}{b}$ is disconnected.

 If $x$ and $y$ are on opposite sides of $C$, then
 $\edge{x}{y}$ crosses      $\edge{a}{b}$,
since the other edges of $C$ have no crossable edges,
such that there is a decomposition. In particular, an ambiguous
separating quadrangle and a 5-clique containing $a,b,x,y$ are excluded.
Then there is an X-quadrangle $Q(a,x,b,y)$.
Otherwise, assume that $x$ and $y$ are in the interior of $C(a,b,c,d)$.
If   $w \neq a,b,c,d$ is a further common neighbor of $x$ and $y$,
then one of $\edge{x}{w}$ and $\edge{y}{w}$ must be a crossable edge,
that is a bridge, which is excluded, since only $\edge{x}{y}$ is a bridge.
Then $N(x) = \{a,b,c,d,y\}$ and $N(y) = \{a,b,c,d,x\}$, so that $G[a,b,c,d,x,y] = D_7$,
which is the graph in Fig.~\ref{fig:quad-E}.
Then $G-C(a,b,c,d)-\mathcal{E}(e)$ decomposes for every edge $e$ in $C(a,b,c,d)$,
such that $\beta_i(C) = (1,1,1,1)$ for $i=1,2,3,4$,
which contradicts the uniqueness of $C(a,b,c,d)$.

At last, suppose edge $e$ of $C(a,b,c,d)$ with $e \neq \edge{a}{b}$ is crossed
by some edge  $\edge{u}{v}$. Then $\edge{u}{v}$ is not a bridge by our assumption on unique separating  quadrangles. The interior (exterior)
of $C(a,b,c,d)$ has at least one vertex, since $C(a,b,c,d)$ has no chords
and $C(a,b,c,d)$ is not a separating 4-cycle.  Then $x$ and $y$ are connected
via $\edge{u}{v}$ in $G''$, since
vertices $x$ and $u$ are connected in $G''$ if they are on the same
side of $C(a,b,c,d)$ and so are $y$ and $v$. By 5-connectivity, there is a path
between $x$ and $u$ in $G$ that avoids $a,b,c,d$, and the path does not contain
edge $\edge{x}{y}$, since $y$ is on the other side of $C(a,b,c,d)$.
Hence, the Lemma holds as stated.
 \qed
\end{proof}

There is a clear decision at a unique separating quadruple, such that the
bridge is crossed and the other edges of the 4-cycle and the X-quadrangle are marked,
and their set of crossable edges is set to $\emptyset$.
A unique separating quadruple is searched by considering all 4-cycles $C$, and then
computing $G_e=G-C- \mathcal{E}(e)$ for every edge $e$ of $C$. For exactly
one edge $e$, graph $G_e$ is disconnected, and if so it
 has two parts that are connected by a single edge.
Observe that a separating quadruple is either
unique or ambiguous or there are subgraphs as shown in Figs.~\ref{fig:quad-ambig}
and \ref{fig:quadrangles}.
\\

It remains to consider separating quadruples with
other bridge pattern and fewer restrictions. For convenience, we  search for them after
  unique separating quadruples.

Assume  an edge $e$ of $C(a,b,c,d)$  is a bridge for two consecutive edges,
such that $e$ is a crossable edge of $\edge{a}{b}$ or of $\edge{b}{c}$, whereas
$e$ is not a crossable edge for an edge of $C(a,b,c,d)$ incident to $d$. For 3-cycles
$C(a,b,c)$ this is excluded, since $K(a,b,c,x,y)$ is a 5-clique. But, separating quadruples have no chord, such that edge $\edge{a}{c}$ is missing.
We have that $G'=G-C(a,b,c,d)-e$ partitions into two parts that are only
connected by $e$. Then there is ambiguous separating quadruple,
which has the bridge pattern $\beta_1(C)=\beta_2(C)=(1,1,0,0)$ and
$\beta_3(C)=\beta_4(C)=($-,-,-,-$)$. Ambiguous separating quadruples
are searched later, after
singular separating triangles and $K_5$-destroyers.

\begin{lemma} \label{lem:single-bridge-in-quad}
Let $C=C(a,b,c,d)$ be a  separating quadruple such that $G''=G-C-\edge{x}{y}$
decomposes. Then $G''$ has two parts $G_x$ and $G_y$ with
$x$ in $G_x$ and $y$ in $G_y$ and $\edge{x}{y}$ is a crossable
edge of an edge in $C(a,b,c,d)$. In addition,

(i) if $\edge{x}{y}$ is a bridge of each edge $e_i$ of  $C$, such that
$\beta_i(C)=(1,1,1,1)$ for $i=1,2,3,4$, then
$G_x$ consists of vertex $x$ of degree five, as shown in Fig.~\ref{fig:quad-A}

(ii) if $\edge{x}{y}$ is a bridge of $\edge{a}{b}$ or $\edge{b}{c}$ and not of
$\edge{a}{d}$ or $\edge{c}{d}$, such that $\beta_1(C)=\beta_2(C)=(1,1,0,0)$
and $\beta_3(C)=\beta_4(C)=($-,-,-,-$)$,
then $C(a,b,c,d)$ is an ambiguous separating quadruple,
as shown in Fig.~\ref{fig:quad-ambig}

(iii) if edge $\edge{x}{y}$ is the single bridge of $\edge{a}{d}, \edge{a}{b}$ and
$\edge{b}{c}$, whereas $\edge{c}{d}$ has    two bridges, that is
$\beta_i(C)=(1,1,1,1)$ for $i=1,2,4$ and $\beta_3(C)=(1,1,2,1)$,
then  $G[a,b,c,d,x,y,z]$ is the graph shown in  Fig.~\ref{fig:quad-E}.

(iv) if edge $\edge{x}{y}$ is the single bridge of $\edge{a}{d}$ and
$\edge{b}{c}$, whereas $\edge{a}{b}$ and $\edge{c}{d}$ have    two bridges, that is
$\beta_2(C)=\beta_4(C)=(1,1,1,1)$  and $\beta_1(C)=\beta_3(C)=(2,1,2,1)$
then  $G[a,b,c,d,x,y,z]$ is the graph shown in  Fig.~\ref{fig:quad-C}.
\end{lemma}

\begin{proof}
Let $C=C(a,b,c,d)=C(v_1,v_2,v_3,v_4)$.
Clearly, $G''$ has exactly two parts containing $x$ and $y$, respectively, since there
is a single bridge. Edge  $\edge{x}{y}$ crosses an edge of $C(a,b,c,d)$, since there is a separating  4-cycle, otherwise. Hence, $x$ and $y$ are on opposite sides of $C(a,b,c,d)$.
There are no chords $\edge{a}{c}$
and $\edge{b}{d}$, so that there is no    5-clique containing three vertices of $C$.
Hence, $\edge{x}{y}$ is a crossable
edge for an edge of $C(a,b,c,d)$.

 For (i), vertices  $a,b,c,d$  are neighbors of $x$  and of $y$, since  $\edge{x}{y}$
is a bridge for each edge of $C(a,b,c,d)$. If $x$ has a further neighbor
$w \not\in\{a,b,c,d,y\}$,
an edge of $C(a,b,c,d)$ is crossed by an edge incident to $w$ due to 5-connectivity
and the absence of separating triples, so that there is another bridge pattern.
Hence, $G_x$ consists of vertex $x$, which has degree five.

For (ii), vertices $a,b,c$ are neighbors of $x$ and $y$ and $d$ is not in $N(x) \cap N(y)$.
There is no $K_5$ containing $a,b,x,y$ or $b,c,x,y$, and also $a$ and $c$ are not
in a $K_5$, since there is no chord $\edge{a}{c}$. Then $\edge{x}{y} \in
\mathcal{E}\edge{a}{b} \cap \mathcal{E}\edge{b}{c}$.  If $z$ is a further neighbor
of $x$, such that $x$ is no a part of $G''$, then $C(a,x,c,d)$ is a separating quadruple,
which is excluded by the assumption.
Hence, $C(a,b,c,d)$ is an ambiguous separating quadruple, which proves (ii).

For (iii) and (iv)  observe that
 $\{a,b,c,d\} \subseteq N(x) \cap N(y)$.
As a single bridge suffices for a decomposition, at most one
edge of $C$ is crossed and vertex $x$ (one of $x$ and $y$) has degree five.
Then only the graphs in Figs.~\ref{fig:quad-A}, \ref{fig:quad-C}
and \ref{fig:quad-E} remain, in which vertex $x$ is isolated for $G''$.
Now the bridge pattern determines the graph.
\qed
\end{proof}

If $C(a,b,c,d)$ is a separating quadruple and one part of $G'=G-C(a,b,c,d)-\mathcal{E}\edge{a}{b}$ is a small graph consisting a single vertex  $x$ or of two vertices
$x$ and $y$, then  the crossed edges are determined by the bridge pattern.
Algorithm $\mathcal{N}$ removes the edge that crosses $C(a,b,c,d)$,
such that thereafter $C(a,b,c,d)$ is a separating 4-cycle (or all crossed edges
of the small component) and marks edges from all detected X-quadrangles.

So far, a single bridge of any edge of $C(a,b,c,d)$  suffices  to decompose graph $G$.
Next, we search for separating quadrangles, such that some edges of $C(a,b,c,d)$
need at least two bridges, whereas there is no decomposition by the removal of
the crossable edges for other edges of $C(a,b,c,d)$.
Then only the graphs in Figs.~\ref{fig:quad-B}, \ref{fig:quad-D} and \ref{fig:quad-F} remain.

\begin{figure}[t]   
  \centering
  \subfigure[] {
     \includegraphics[scale=0.5]{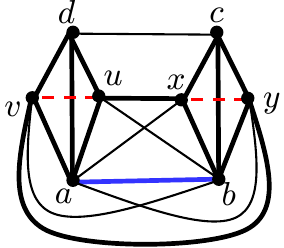}  
    \label{fig:sepquad1}
  }
  \subfigure[] {
      \includegraphics[scale=0.5]{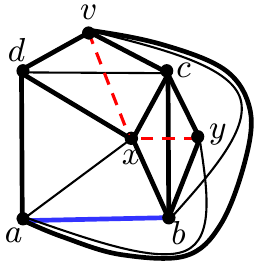}  
      \label{fig:sepquad2}
  }
  \subfigure[] {
      \includegraphics[scale=0.5]{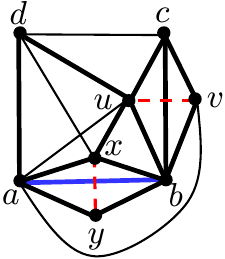}
      \label{fig:sepquad3}
  }
  \subfigure[] {
      \includegraphics[scale=0.5]{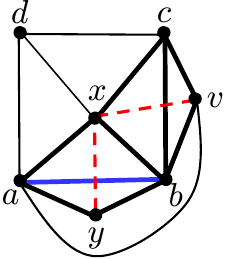}
      \label{fig:sepquad4}
  }
  \subfigure[] {
      \includegraphics[scale=0.5]{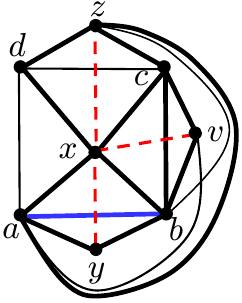}
      \label{fig:sepquad5}
  }
  \caption{Illustration for the proof of Lemma~\ref{lem:two-bridges-quad}.
Separating quadruples with two bridges for edge $\edge{a}{b}$.
(a) Edge $\edge{x}{y}$ crosses $\edge{b}{c}$ and $\edge{u}{v}$ crosses $\edge{a}{d}$,
 (b)  $\edge{x}{y}$ crosses $\edge{b}{c}$, $\edge{u}{v}$ crosses $\edge{c}{d}$ and $u=x$.
(c) Edge $\edge{a}{b}$ is crossed. The second bridge $\edge{u}{v}$ crosses $\edge{b}{c}$.
There are no further vertices or edges inside $C(a,b,c,d)$ if there are no smaller separators. $C(a,d,c,v)$ is a unique separating quadruple.
(d) Case $x=u$. (e) Three bridges and a separating triple $C(a,d,z)$.
  }
  \label{fig:two-bridges-quad}
\end{figure}


\begin{lemma} \label{lem:two-bridges-quad}       
Suppose  $G'=G-C(a,b,c,d)- \mathcal{E}\edge{a}{b}$ is disconnected
but $G$ does not decompose if $C(a,b,c,d)$ and a single edge are removed.
 Then $\edge{a}{b}$ is crossed by one of its bridges.
\end{lemma}

\begin{proof}
At least two edges of $C=C(a,b,c,d)$ are crossed, since
at least one edge of $C$ is crossed  by 5-connectivity and $G$ does not decomposes by a single edge by assumption.

For a contradiction, suppose that $\edge{a}{b}$ is uncrossed.
Then $\edge{b}{c}$ is crossed by a bridge $\edge{x}{y}$ of $\edge{a}{b}$.
If also $\edge{a}{d}$ is crossed by a bridge $\edge{u}{v}$ of $\edge{a}{b}$,
as shown in Fig.~\ref{fig:sepquad1},
then there is a separating 4-cycle $C(v,c,d,y)$, since its edges are uncrossed.
The case in which $u=x$ is similar.
If bridges $\edge{x}{y}$ and $\edge{u}{v}$ of $\edge{a}{b}$ cross
$\edge{b}{c}$ and $\edge{c}{d}$, respectively, that is two consecutive edges of $C$
edges $\edge{a}{y}$ and $\edge{b}{v}$ cross, such that there is a separating 3-cycle
$C(a,v,d)$, contradicting the assumption.
\qed
\end{proof}

In consequence, there is a clear decision and a unique 1-planar embedding
if the requirements for Lemmas~\ref{lem:single-bridge-in-quad}
and \ref{lem:two-bridges-quad} are fulfilled. In fact, a simple case analysis shows that
two bridges of $\edge{a}{b}$ cannot cross edge $\edge{a}{b}$ and the adjacent edge
$\edge{b}{c}$ in the absence of smaller separators, in
particular unique separating quadrangles, as Figs.~\ref{fig:sepquad3}-\ref{fig:sepquad5}
illustrate. Hence, only the graphs
in Figs.~\ref{fig:quad-B}, \ref{fig:quad-D} and \ref{fig:quad-F}
remain if two or three bridges are necessary to decompose the graph.
In any case, algorithm $\mathcal{N}$
removes the edges that cross edges from the 4-cycle $C(a,b,c,d)$ and later recognizes
$C(a,b,c,d)$ as a separating 4-cycle.\\

\noindent \textbf{Remark 1}
Note that the algorithm by Chen et al.~\cite{cgp-rh4mg-06} fails on separating
quadruples with two or three bridges if it chooses a bridge
that does not cross $\edge{a}{b}$. Then it marks some wrong edges, such that
 the final planarity test fails,
although the graph   is a \emph{T1P} graph. For example, it may classify edge
$\edge{x}{y}$ in Fig.~\ref{fig:quad-C} or \ref{fig:quad-E} as crossed.

Second, there is an error for circle  generalized two-stars $cG2S_k$
or $sG2S_k$ in~\cite{cgp-rh4mg-06}, see Fig.~\ref{fig:DP}. This special case has
been discarded by Chen et al.
If $p$
and $q$ are the poles and $v_1,\ldots, v_k$ the other vertices of
$cG2S_k$, then there is a failure if the algorithm first chooses
e.g., edge $\edge{p}{v_8}$ to be crossed and later $\edge{q}{v_4}$
and marks edges $\edge{p}{v_7}$ and $\edge{q}{v_5}$. Then
$G[p,q,v_4,\ldots,v_8]$ does not admit a 1-planar drawing with the
edges of $C(p,v_4,q, v_8)$ and $\edge{p}{v_7}$ and $\edge{q}{v_5}$
uncrossed.  With such a choice, algorithm $\mathcal{A}$
 rejects $G$, so that $\mathcal{A}$ is false positive.

\subsection{Separating Triangle} \label{sect:triangle}

Suppose   $C(a,b,c)$ is a   separating triangle such that
$G'=G-C(a,b,c)- \mathcal{E}\edge{a}{b}-\mathcal{E}\edge{b}{c}$ is disconnected.
 Intuitively, the following is clear.
If $u$ and $v$ are vertices on opposite sides of $C(a,b,c)$, then
there are five vertex disjoint  paths between   $u$ and $v$ that pass through one of the vertices $a,b,c$ or through one of the two edges crossing $C(a,b,c)$. In consequence, $G'$ has two
parts, one on either side of $C(a,b,c)$ if the separating triangle decomposes along $C(a,b,c)$.
However, this may not be true. Two edges of $C(a,b,c)$ and their crossable edges must be
chosen appropriately. Otherwise, an edge and a bridge
may not cross, as edges $\edge{x}{y}$ and $\edge{x}{w}$ in Fig.~\ref{fig:counterex-triple1}, or all three edges of $C(a,b,c)$
are crossed, as in Fig.~\ref{fig:counterex-triple2}.
Chen et al.~\cite{cgp-rh4mg-06} (Lemma 7.15) circumvent this problem and search for separating
triangles after searching separating quadrangles and strongly separating triangles. Note that
vertex $y$ is the center of a strongly separating triangle in Figs.~\ref{fig:counterex-triple1}
and  \ref{fig:counterex-triple2}.


We proceed in a different way and specialize separating triangles, such that
there are two independent bridges or  the separating triangle is   singular or  ambiguous.
Thereby, we can relax the use of the generalized separators,
such that separating quadrangles and (strong) separating triangles can be searched
in any order, whereas ambiguous separating quadrangles and triangles are searched last.
This is due to the  the absence of  the respective separators as a requirement
in the  proofs, as opposed to \cite{cgp-rh4mg-06}.
At last, we use separating triangles to destroy all remaining 5-cliques.\\

\begin{figure}[t]   
  \centering
     \subfigure[] {
     \includegraphics[scale=0.45]{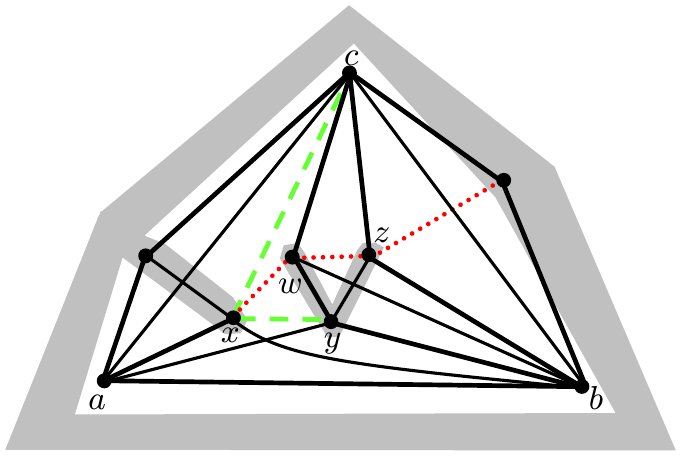}  
    \label{fig:counterex-triple1}
  }
  \hfil
  \subfigure[] {
      \includegraphics[scale=0.45]{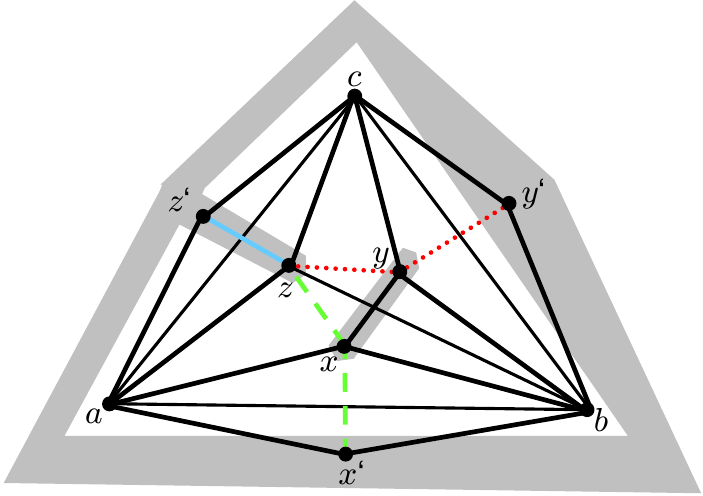} 
      \label{fig:counterex-triple2}
  }
  \label{fig:counterexample-triple}
\caption{
A separating triangle $C(a,b,c)$ such that $G'=G-C(a,b,c)-
\mathcal{E}\edge{a}{b}-\mathcal{E}\edge{b}{c}$ decomposes into two parts, drawn shared.
 Crossable edges of $\edge{a}{b}$
are drawn dashed and green and crossable edges of $\edge{b}{c}$ are drawn dotted and red.
Bridges $\edge{x}{y}$ and $\edge{x}{w}$ are  not independent. Vertex $y$ is the center
of a strongly separating triangle and $w$ is the center of a strongly separating quadrangle.
(b) Each edge of $C(a,b,c)$ is crossed, but the bridges, drawn dashed and green and
dotted and red,   are not
indpendent. Vertices $x$ and $y$ are a center of a strongly separating triangle.
\newline
Algorithm $\mathcal{A}$ by Chen et al.~\cite{cgp-rh4mg-06} finds the edge that crosses
$\edge{b}{c}$
by a separating quadrangle $C(b,c,w,y)$ and $C(b,c,z,x)$, respectively.
Algorithm $\mathcal{N}$ finds a bridge-independent separating triangle with edges
$\edge{a}{c}$ and $\edge{b}{c}$ or the separating quadrangle $C(b,c,w,y)$ for the
graph in (a), which  ever comes first,  and similarly for the graph in (b) if edge
$\edge{z}{z'}$, drawn blue, is omitted. It uses the unique independent quadrangle
$C(b,c,z,x)$, otherwise, and then the separating quadrangle $C(b,c,z,x)$.
}
\end{figure}

Let $C(a,b,c)$ be a 3-cycle such that   $\edge{a}{b}$ is crossed by $\edge{x}{y}$,
$\edge{b}{c}$ is crossed by $\edge{u}{v}$, and the third edge $\edge{a}{c}$ is uncrossed
in a 1-planar embedding $\Gamma(G)$, as shown in Fig.~\ref{fig:triangle}.
Then $C(a,b,c)$ is a separating triangle.
The edges
$\edge{x}{y}$ and $\edge{u}{v}$ are independent by 5-connectivity if both parts
of $G'=G-C(a,b,c)-\edge{x}{y}-\edge{u}{v}$ contain
at least two vertices. Otherwise, $C(a,b,c)$ is a
strongly separating triangle or $C(a,x,c)$ is a separating 3-cycle if $x=u$.
The  decomposition of $G$  is enforced by a bridge-independent separating triangle, where
 $G'=G-C(a,b,c)-
\mathcal{E}\edge{a}{b}-\mathcal{E}\edge{b}{c}$ partitions into nontrivial parts and
edges $\edge{x}{y} \in \mathcal{E}\edge{a}{b}$ and $\edge{u}{v} \in \mathcal{E}\edge{b}{c}$
are  independent bridges, that is
none of $x, y, u, v$  is not incident to another  bridge in $\mathcal{E}\edge{a}{b} \cup   \mathcal{E}\edge{b}{c}$.

In the proofs, we use the following observation. If $\edge{x}{y}$ is a bridge
and $w$ is a neighbor of $x$ and $y$ with $w \neq a,b,c$, then $\edge{x}{w}$
or $\edge{y}{w}$ is a crossable edge of one of the edges of  $C(a,b,c)$. Otherwise,
$x$ and $y$ are connected via $w$. However, these properties can only be fulfilled
by graph $H_7$, shown in Figs.~\ref{fig:quad-A} and \ref{fig:small-H7},
respectively.

\begin{lemma} \label{septriangle-independentbridges}  
Let $C(a,b,c)$ be a 3-cycle such that $\edge{a}{c}$ is uncrossed,
  $\edge{a}{b}$ is crossed by $\edge{x}{y}$ and
$\edge{b}{c}$ is crossed by $\edge{u}{v}$,
in a  T1P  embedding $\Gamma(G)$, such that vertices $x,y,u$ and $v$ are distinct
and there is no 5-clique containing $a,b,x,y$ and $b,c,u,v$, respectively.
Then $C(a,b,c)$ is a bridge-independent separating triangle  or
one component is graph $H_7$.
\end{lemma}
\begin{proof}
Clearly, $G''=G-C(a,b,c)-\edge{x}{y}-\edge{u}{v}$ decomposes  such that
$\edge{x}{y}$ and $\edge{u}{v}$ are bridges of $G''$. Then also
$G'=G-C(a,b,c)-\mathcal{E}\edge{a}{b}-\mathcal{E}\edge{b}{c}$ is disconnected,
since 5-cliques are excluded.

It remains to show that $G'$ has two parts or there  graph $H_7$,
which is a 6-clique without two independent edges, see
Figs.~\ref{fig:quad-A} and \ref{fig:small-H7}.
A further partition
of $G'$ is due to further crossable edges.
We claim that the parts of $G'$ and $G''$ coincide up to crossable edges of $\edge{a}{b}$
and $\edge{b}{c}$.  Therefore,
consider a crossable edge $\edge{x'}{y'}$ of $\edge{a}{b}$ with $\edge{x'}{y'} \neq \edge{x}{y}$.
Since $\edge{a}{b}$ is crossed by $\edge{x}{y}$, the
4-clique $K(a,b,x',y')$ is represented as a tetrahedron. It cannot be represented as an
X-quadrangle, since $a$ and $b$ are on opposite sides of $Q(a,x,b,y)$. Then
$x'$ and $y'$ are on the same side of $C(a,b,c)$, since, otherwise, edge $\edge{x'}{y'}$
crosses an edge of $C(a,b,c)$ incident to $c$ and so creates the 5-clique $K(a,b,c,x',y')$.

We claim that  $x'$ and $y'$ are in the same part of $G'$.
Suppose that $x'$ is the center of the tetrahedron for  $K(a,b,x',y')$ and assume
that $x'$ and $y'$ are in the interior of $C(a,b,c)$.
By 5-connectivity, there is a path
between $x'$ and $y'$ that passes through a common neighbor $w'$  with $w' \neq a,b,c$.
The  $\edge{x'}{w'}$ or $\edge{y'}{w'}$ must be a crossable edge, otherwise it establishes
a connection between $x'$ and $y'$ in $G'$. A crossable edge of $\edge{a}{b}$ is excluded,
since that creates the 5-clique $K(a,b,x',y',w')$. An edge incident to $c$ does not cross
$\edge{b}{y'}$, since it creates an X-quadrangle such that $\edge{b}{c}$ is uncrossed,
contradicting the assumption.
Hence, $w'$ is outside $C(a,b,y')$, such that edges $\edge{b}{y'}$ and $\edge{x'}{w'}$ cross.
Then edge $\edge{x'}{b}$ is uncrossed. Also edges $\edge{a}{x'}$ and $\edge{a}{y'}$
are uncrossed, since there are no vertices inside $C(a,x',y')$, such that $C(a,x',y')$
is a triangle. Otherwise, either $x'$ and $y'$ are connected via a vertex inside
$C(a,x',y')$ in $G'$ or there is a crossable edge of $\edge{b}{c}$ that creates a 5-clique
containing $a,b,x'$ and $y'$, a contradiction.
Now $C(a,b,x')$ is a separating triple if $x \neq x'$. Hence, we have $x=x'$.
Similarly, there is no vertex $z$ inside $C(c,w',y')$, otherwise, either $x'$ and $y'$
are connected via   $z$ in $G'$ or $z$ is a neighbor of $y'$ and $w'$ and is connected
to $b$ and $c$, since $\edge{z}{y'}$ or $\edge{z}{w'}$ must be a crossable edge of
$\edge{b}{c}$ to avoid the connection, such that $K(b,c,y',w', z)$ is a 5-clique,
a contradiction.
Then edges $\edge{c}{y'}$ and $\edge{c}{w'}$ are uncrossed.
Since separating triples are excluded, $C(a,c,y')$ is a triangle and also $C(c,y',w')$ is
a triangle with uncrossed edges in the boundary. Then $C(b,c,w')$ is a triangle, so that
$w'=u$, and the subgraph $G[a,b,c,x',y',w']$ in the interior of $C(a,b,c)$ is $H_7$,
where edges $\edge{a}{w'}$ and $\edge{c}{x}$ are missing from $K_6$.
Graph $G'$ partitions into vertex $y$, the edge $\edge{x}{u}$ and a part
in the exterior of $C(a,b,c)$.
The  case for  $\edge{b}{c}$ is similar.
Hence, either each of $x,y,u$ and $v$ is incident to a single bridge, namely
$\edge{x}{y}$ and $\edge{u}{v}$, or there is $H_7$.

A 4-clique $K(a,b,x',y')$ cannot be represented as an X-quadrangle, since
edge $\edge{a}{b}$ is already crossed by $\edge{x}{y}$. Another X-quadrangle implies
that $\edge{a}{b}$ is uncrossed, a contradiction.
\qed
\end{proof}

\begin{lemma} \label{septriangle-twobridges}     
Let $C(a,b,c)$ be a bridge-independent separating triangle.
Then $\edge{x}{y}$ crosses $\edge{a}{b}$, and $\edge{u}{v}$ crosses
$\edge{b}{c}$, whereas $\edge{a}{c}$ is uncrossed.
\end{lemma}

\begin{proof}
If $x$ and $y$ are on the same side of $C(a,b,c)$, then there is a common
neighbor $w \neq a,b,c$ so that $\edge{x}{w}$ or $\edge{y}{w}$ is a crossable edge.
Then it is a crossable edge of $\edge{b}{c}$, since there is  a 5-clique
$K(a,b,x,y,w)$, otherwise. Then $\edge{x}{y}$ and the bridge
among $\edge{x}{w}$ or $\edge{y}{w}$ are not independent, contradicting the assumption.
If $x$ and $y$ are on opposite sides of $C(a,b,c)$, then it does not cross an edge
of $C(a,b,c)$ incident to $c$, since $K(a,b,c,x,y)$ is a 5-clique. Hence,
the case where $\edge{x}{y}$ crosses $\edge{a}{b}$ remains.
The case for $\edge{b}{c}$ and $\edge{u}{v}$, is similar.
At last, suppose edge $\edge{a}{c}$ is crossed by an edge $\edge{x'}{y'}$,
such that $x$ and $x'$ are inside $C(a,b,c)$. We claim that $x$ and $x'$
are connected in $G'$ if $C(a,b,c)$ is a bridge-independent separating triangle.
To see this, consider the neighbors of $a$
from $b$ to $c$ in the interior of $C(a,b,c)$, that is the sequence
$b, w_1,\ldots, w_r, c$  such that $x=w_1, x'=w_r$ and $r \geq 1$.
Then there are vertices $w_i$ and $w_{i+1}$ such that edge $\edge{w_i}{w_{i+1}}$
is  crossable. Then $i>1$, since there are two bridges incident to $x$, otherwise.
Assume that $i$ is minimal.

If $\edge{w_i}{w_{i+1}}$ is a crossable edge of $\edge{b}{c}$, then there is a 5-clique
$K(a,b,c,w_i, w_{i+1})$. If it is a  crossable edge of $\edge{a}{b}$, then
$w_i$ has a neighbor $z \not\in \{a,b,w_{i-1}, w_{i+1}\}$, since a separating
triple $C(a,b,w_i)$ is excluded. If $z=c$, then $K(a,b,c,w_i, w_{+1})$ is a 5-clique,
violating the crossability of $\edge{w_i}{w_{i+1}}$, or $w_{i+1}$ is
incident to two bridges, namely $\edge{w_i}{w_{i+1}}$ and $\edge{w_{i+1}}{z}$,
contradicting bridge-independence.
\qed
\end{proof}

Hence, there is a clear decision at bridge-independent separating triangles, such that
algorithm $\mathcal{N}$   removes the bridges,   marks the edges of $C(a,b,c)$
and updates the set with pairs   $(e, \mathcal{E}(e))$.
In the next step, $C(a,b,c)$ is
recognized as a separating 3-cycle.

\subsection{Strongly Separating Triangles} \label{sect:stronglyseparatingtriangles}
Next we consider strongly  separating triangles and distinguish between
singular and
ambiguous ones, which differ by an edge between the vertices of the bridges.
The proofs need 5-connectivity, but not the exclusion of
 bridge-independent separating triangles or unique
separating quadrangles. Hence, there is  no
absolutely necessary ordering among these generalized separators.

The following is obvious, since there are many uncrossed edges.

\begin{lemma} \label{singular septriangle}  
Let $C(a,b,c)$ be a 3-cycle such that $\edge{a}{c}$ is uncrossed,
  $\edge{a}{b}$ is crossed by $\edge{d}{x}$ and
$\edge{b}{c}$ is crossed by $\edge{d}{y}$ in $\Gamma(G)$.
Then $C(a,b,c)$ is a strongly  separating triangle   with a center $d$,
which has   degree five if there is no edge $\edge{c}{x}$ or $\edge{a}{y}$
Moreover, there  is a singular separating triangle
if there is no edge $\edge{x}{y}$ and an ambiguous separating triangle if $\edge{x}{y}$
exists.
\end{lemma}

\begin{proof}
Edges $\edge{a}{c}, \edge{a}{d}, \edge{b}{d}$ and $\edge{c}{d}$ are uncrossed
in $\Gamma(G)$.
Since there is no separating 3-cycle, $C(a,c,d)$ is a triangle. Then $d$ has degree five,
so that it is isolated by the removal of $C(a,b,c)$ and edges  $\edge{d}{x}$ and
 $\edge{d}{y}$.

Edge  $\edge{d}{x}$ is a crossable edge of $\edge{a}{b}$
if $G[a,b,c,d,x]$  is no 5-clique, which holds if edge $\edge{c}{x}$ is missing,
and similarly for  $\edge{d}{x}$. Then $C(a,b,c)$ is a strongly separating
triangle with center $d$.
Moreover, there  is a singular separating triangle
if there is no edge $\edge{x}{y}$ and an ambiguous separating triangle if $\edge{x}{y}$
exists.
\qed
\end{proof}

 \begin{lemma} \label{lem:singular-representation1}  
Suppose  $G[a,b,c,d,x,y]$  consists of three 4-cliques
$K(a,b,c,d), K(a,b,d,x)$ and $K(b,c,d,y)$  and vertex $d$ has degree five.
Then $K(a,b,c,d)$ is represented as a
 tetrahedron with center $d$ and  there are X-quadrangles $Q(a,d,b,x)$ and $Q(b,d,c,y)$.
  \end{lemma}

\begin{proof}
Note that there is on 5-clique containing one of the 4-cliques and that edge $\edge{x}{y}$
is missing. We show that the representation of $K(a,b,c,d)$ as a tetrahedron
with center $a$ (or $c$) or $b$ or as an X-quadrangle implies that edge $\edge{x}{y}$
exists or there is a 5-clique containing one of the 4-cliques.

If $K(a,b,c,d)$ is a tetrahedron with center $a$, then edges $\edge{a}{x}$ and $\edge{b}{d}$
cross such that there is an X-quadrangle $Q(b,a,d,x)$
To see this,  $x$ cannot be placed in
$C(a,b,d)$ such that there is a tetrahedron, since then edges $\edge{d}{b}$
and $\edge{d}{c}$
are uncrossed and $C(b,c,d)$ is a separating triple.
There is a 5-clique if edge $\edge{a}{x}$ crosses an edge incident to $c$.
Hence, there is an X-quadrangle $Q(b,a,d,x)$ and $x$ is outside $(b,c,d)$.
The case for $K(b,c,d,y)$ is similar, Then $x$ and $y$ are consecutive neighbors
at $d$, such that there is edge $\edge{x}{y}$ by the triangulation.
The case for a center $b$ or $c$ is similar.

Next,  suppose that there is an X-quadrangle $Q(a,b,c,d)$ or $Q(a,b,d,c)$, such that
$\edge{b}{d}$ or $\edge{a}{d}$ is crossed. The case where $\edge{c}{d}$ is crossed is similar. Now $K(a,b,d,x)$ and $K(b,c,d,y)$ are represented as a tetrahedron.
In the first case, vertices $x$ and $y$ are consecutive neighbors at $d$.
In the second case, and edge incident to $x$ crosses an edge incident to $y$,
so that in any case there is edge $\edge{x}{y}$, contradicting the assumption.
Hence, only the representation as claimed remains.
\qed
\end{proof}

 \begin{lemma} \label{lem:singular-septriangle}   
If  $C(a,b,c)$ is a singular separating triangle with
  center $d$, then $d$ has degree five.
\end{lemma}

\begin{proof}
By assumption, graph $G' = G-C(a,b,c)-\mathcal{E}\edge{a}{b}-\mathcal{E}\edge{b}{c}$
consists of exactly two parts,
a center $d$ and a remainder with vertices $V-\{a,b,c,d\}$, which is connected.
There are 4-cliques $K(a,b,d,x)$ and $K(b,c,d,y)$, so that
 $a,b,c,x$ and $y$ are neighbors of $d$.

If $d$ has a sixth neighbor $z$, then
 $\edge{d}{z}$ is a crossable edge of
$\edge{a}{b}$ or $\edge{b}{c}$, since otherwise $z$ and $d$ are in the same part of $G'$.  Then there are edges $\edge{a}{z}$
and $\edge{b}{z}$, such that there are three 3-sharing 4-cliques
$K(a,b,d,c), K(a,b,d, x$) and $K(a,b,d, z)$ if $\edge{d}{z}$ is a crossable edge of
$\edge{a}{b}$.  The case for $\edge{b}{c}$ is symmetric. Now the addition of
vertex $c$ and edges $\edge{c}{a}, \edge{c}{b}$ and $\edge{c}{d}$ enforces a crossing
of an edge incident to $x$ and a 5-clique containing $a,b,d,x$ or $a,b,d,z$,
contradicting the crossability of $\edge{d}{x}$ or $\edge{d}{z}$.
\qed
\end{proof}

 From Lemmas~\ref{lem:singular-representation1}  and \ref{lem:singular-septriangle}
we obtain a unique representation for a singular separating triangle, which is
as assumed and illustrated in Fig.~\ref{fig:triangle} and stated by Chen et al.
\cite{cgp-rh4mg-06} (Lemma 7.14).

Hence, there is a clear decision at singular separating triangles,
such that algorithm $\mathcal{N}$ removes the bridges
$\edge{d}{v_4}$ and $\edge{d}{v_5}$ and marks the remaining edges of the
X-quadrangles containing the bridges.
Note that a strong separating triple with a circle $C(v_1, v_2,
v_3)$ and vertex $d$ of degree four and an X-quadrangle $Q(v_1, d,
v_2, v_4)$  is a separating edge $\edge{v_1}{v_2}$ that crosses
either $\edge{d}{v_3}$ or $\edge{d}{v_4}$.

\begin{lemma} \label{lem:ambiguous}  
Let $d$ be  a vertex of degree five with neighbors $v_1, \ldots,
v_5$. If $G[d, v_1, \ldots, v_5]$ consists of four maximal 4-cliques
$K_1=K(d, v_1, v_2, v_3), K_2=K(d, v_1, v_2, v_4), K_3=K(d, v_2,
v_3, v_5)$ and $K_4=K(d, v_2, v_4, v_5)$,
that is vertices $v_2$ and $d$ are in all four quadrangles
and the vertices are in exactly two. Then $C(v_1, v_2, v_3)$ is
an ambiguous separating triangle,
and conversely.

There are  two T1P embeddings with $K_1$ and $K_3$ as a tetrahedron
and $K_2$ and $K_4$ as an X-quadrangle, or conversely.
\end{lemma}

\begin{proof}
There are two strongly separating triples $C(v_1, v_2, v_3)$ and
$C(v_1, v_2, v_4)$, which each isolate vertex $d$. Hence, there is
an ambiguous separating triangle. Conversely, there are four
4-cliques if there is an ambiguous separating triangle.

By arguments as in the previous proof,
 the  complementary 4-cliques $(K_1, K_3)$ and $(K_2,K_4)$ 
are either both represented as a tetrahedron with center $d$ or as
an X-quadrangle, and vice versa. In any case, edges $\edge{d}{v_1},
\edge{d}{v_2},\edge{d}{v_3}, \edge{v_1, v_3}$ and $\edge{v_1}{v_4}$
are uncrossed and edges $\edge{v_1}{v_2}$ and $\edge{d}{v_5}$ are
 crossed. The other edges are crossed in one embedding and
uncrossed in another.
\qed
\end{proof}

\begin{lemma} \label{lem:triangle2quadruple}  
If $C(a,b,c)$ is an ambiguous separating triangle with center $d$ and
bridges $\edge{d}{x}$ and $\edge{d}{y}$ for $\edge{a}{b}$ and $\edge{b}{c}$,
respectively, then $C(a,x,y, c)$ is an ambiguous separating quadruple.
\end{lemma}

\begin{proof}
From Lemma~\ref{lem:ambiguous}, we obtain that $C(a,x,y,c)$ is an induced
4-cycle that separates vertices $b$ and $d$ from the remainder if $C(a,x,y,c)$
and edge $\edge{b}{z}$ for a further  vertex $z \neq a,b,c,d,x,y$ are removed.
Edge $\edge{b}{z}$ may  cross $\edge{c}{y}$ or $\edge{x}{y}$, so that
$C(a,x,y,c)$ is ambiguous separating quadruple.
\qed
\end{proof}

An ambiguous separating triangle is detected at a vertex $d$ of degree
five that is part of four 4-cliques. Then algorithm $\mathcal{N}$
sets   $\lambda(\edge{v_1}{v_2})=2$ if edges $\edge{d}{v_1},
\edge{d}{v_2},\edge{d}{v_3}, \edge{v_1, v_3}$ and $\edge{v_1}{v_4}$
are unmarked. Before, we have $\lambda\edge{v_1}{v_2}=1$, since
edge $\edge{v_1}{v_2}$ is unmarked, such that its label has not been
changed after the initialization. Otherwise, there is at most one
\emph{T1P} embedding if at least one of these edges is marked, so
that the labels are unchanged. Algorithm $\mathcal{N}$ removes the
crossed edges $\edge{v_1}{v_2}$ and $\edge{d}{v_5}$, and marks the
remaining edges of the X-quadrangles with $\edge{v_1}{v_2}$ and $\edge{d}{v_5}$.
If ambiguous separating quadruples are searched first, then there is
no ambiguous separating triple left by Lemma~\ref{lem:triangle2quadruple}.
\\

Recall that a vertex of degree four defines a $K_5$-triple if it is   an inner vertex of a 5-clique in a
4-connected graph, see Lemma~\ref{lem:create-h2S}. We continue with vertices of degree
five.
Assume that graphs are 5-connected and have no bridge-independent, singular or
ambiguous separating triangle and no unique separating quadrangle.
Consider a 5-clique $K(a,b,c,x,y)$.  Then all
vertices have degree at least five.  In fact, the outer vertices
have degree at least six, since there is a unique separating quadrangle,
otherwise. Also, inner vertices can have degree six, as the handle
generalized two-stars show.
Suppose there is 5-clique $K=K(a,b,c,u,x)$ with two vertices $u$ and
$x$ of degree five in $G$. Then $u$ and $x$ are inner vertices and
$C(a,b,c)$ is the outer boundary of $K$. There are no other vertices
in the interior of $K$, since that increases the degree of $u$ or
$x$ by the triangulation. Vertices $u$ and $x$ are each incident to
an edge that crosses an outer side of $K$, since $G$ is 5-connected,
see Fig.~\ref{fig:triple}.
Vertex $u$ is the corner of two X-quadrangles, which however,
can be a part of 5-cliques, and similarly for  $x$. If the X-quadrangles are (maximal)
4-cliques, then $C(a,b,c)$ is a separating triangle, such that
the bridges are removed and  $C(a,b,c)$  is a separating
3-cycle thereafter. The representation of $K$ is unique and is determined by
algorithm $\mathcal{N}$.

Suppose  a 5-clique $K$ and a 4-clique $K'$ are 3-sharing, that is,
they intersect in a 3-cycle (or $K_3$). Let $K=K(v_1,\ldots, v_5)$
and $K'=K(v_1, v_2, v_3, v_6)$. Then $K'$ is an X-quadrangle $Q(v_1,
v_3, v_2, v_6)$ such that $v_1$ is the top of $\Gamma(K)$ and
$\edge{v_1}{v_2}$ is an outer side. Otherwise, $C(v_1, v_2, v_4)$ is
a separating triple or a separating triangle  that separates $v_6$
and $v_5$.\\

Hence, there is a clear decision if there is an singular separating triangle.
Algorithm $\mathcal{N}$ detects the center $d$ of degree five,
it removes the bridges $\edge{d}{x}$ and $\edge{d}{y}$
and marks the remaining edges of $G[a,b,c,d,u,v]$. In the next step, it will detect
the separating 3-cycle $C(a,b,c)$, such that one  component is $G-d$
with marked edges for $C(a,b,c)$.\\

Next we destroy all 5-cliques.  We consider 5-cliques by the  degree of
the inner vertices. As shown before, both inner vertices have degree four in 3-connected
graphs, such that the outer boundary of the $K_5$ is a separating 3-cycle.
If a $K_5$ has an inner vertex of degree four in a 4-connected graph, then   either there is a separating triple
such that the bridge crosses one outer side of the 5-clique or
there is another $K_5$, such that there is an ambiguous separating
triple, which ultimately leads to $hG2S_k$, as shown in Lemma~\ref{lem:create-h2S}.

The  inner triangles of a $K_5$ are cleared by separating triples and bridge-independent and
singular
separating triangles. Each outer side is crossed by an edge
that is incident to an inner vertex if the graphs are 5-connected.
Then all vertices have degree at least five and the inner
vertices have degree exactly five.
If both outer sides of a 5-clique are crossed such that there are $K_4$,
then the crossing edges are detected by a bridge-independent separating triangle.
Thereafter, there is a separating 3-cycle such that $K_5$ is one component.
Let $C(a,b,c)$ be the outer boundary of a $K_5$ with inner vertices
$x$ and $y$.
If both outer sides of a 5-clique are crossed by edges from 5-cliques,
that is there are 5-cliques $K(a,b,c,x,u)$ and $K(a,b,c,y,v)$, then two edges
incident to $u$ and $v$ cross outside $C(a,b,c)$, such that the
obtained graph is not a subgraph of a 4-connected \emph{T1P} graph.
Hence, there remains the case of a $K_4$ and a $K_5$.

\begin{lemma}  \label{lem:5-cliques}     
Let $Q=G[U]$ be a 5-clique with $U= \{a,b,c,x,y\}$, such that
vertices $x$ and $y$ have degree five and $a,b,c$ have degree at
least six. If there are a 5-clique $G[a,b,c,y,v]$ and a 4-clique
$G[a,c,x,u]$ with vertices $u, v  \not\in U$, then $Q(a,x,c,u)$ is an
X-quadrangle. Then there is an ambiguous
separating triple if edge $\edge{x}{u}$ is removed.
\end{lemma}

\begin{proof}
The inner and outer vertices of the \emph{T1P} embedding
$\Gamma(G)$ are distinguished by
their degree, such that  $x$ and $y$ are inner  and $a,b,c$ are
outer vertices. Vertex $b$ is in the 5-clique and not in the
4-clique with vertices $u, v \not\in U$, which distinguishes $b$
from $a,c$. By the absence of smaller separators, edge $\edge{x}{u}$
from the 4-clique $K(a,c,x,u)$ must cross $\edge{a}{c}$.
After the removal of  $\edge{x}{u}$, there is a 5-clique with
vertices $a,b,c,x,y$ and vertex $x$ of degree four,
such that there is a separating triple.
\qed
\end{proof}

Hence, we have reduced the case of a 5-clique with two vertices
of degree five to the case of  a 5-clique with a vertex of degree four.
If edge $\edge{x}{u}$ is removed from $G$ and an ambiguous separating
triple is the smallest generalized separator of $G-\edge{x}{u}$,
then $G-\edge{x}{u} = hG2S_k$ by Lemma~\ref{lem:create-h2S}.
In consequence, $G = xG2S_k$ such that  algorithm $\mathcal{N}$
will detect directly that $G$ is a generalized two-star.\\

\begin{figure}[t]   
  \centering
\subfigure[] {
     \includegraphics[scale=0.5]{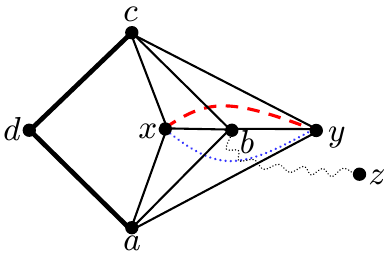}  
    \label{fig:ambig-quad}
  }
  \hfil
  \subfigure[] {
      \includegraphics[scale=0.5]{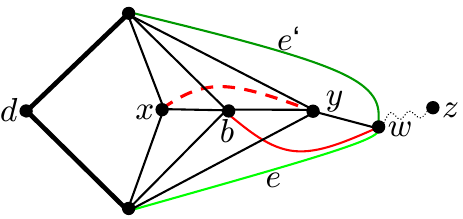} 
      \label{fig:quad-ambig2}
  }
  \caption{Illustration for the proof of
  Lemma~\ref{lem:createcG2S}.
  (a) An ambiguous separating quadrangle $C(a,b,c,d)$ with   bridge $\edge{x}{y}$
  that crosses $\edge{a}{b}$, drawn dotted and blue or
  $\edge{c}{b}$, drawn dashed and red. Only one of them exists.
  There is a path from $b$ to $z$.
  (b) A next ambiguous separating quadrangle $C(a,y,c,d)$ with   bridge $\edge{b}{w}$
  }
  \label{fig:qconstructc2S_n}
\end{figure}

\begin{figure}[t]   
  \centering
     \subfigure[] {
     \includegraphics[scale=0.55]{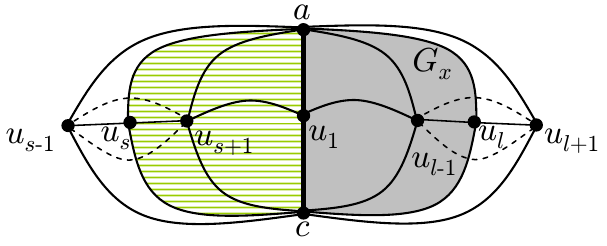}  
    \label{fig:createcG2Sa}
  }
  \subfigure[] {
      \includegraphics[scale=0.55]{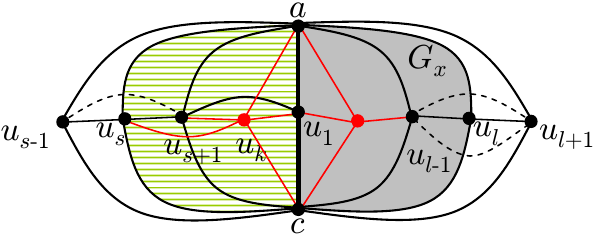} 
      \label{fig:createcG2Sb}
  }
  \caption{Illustration for the proof of
  Lemma~\ref{lem:createcG2S}. Steps to create $cG2S_k$.
  }
  \label{fig:createcG2S}
\end{figure}

Let's return to graphs with ambiguities but without 5-cliques in generalized separators.

\begin{lemma} \label{lem:createcG2S}  
  If $G$ has an ambiguous separating quadrangle, then $G=cG2S_k$ or $G=sGS_k$.
\end{lemma}

\begin{proof}
Let $C(a,b,c,d)$ be an ambiguous separating quadrangle, such that
the bridge $\edge{x}{y}$   crosses either $\edge{a}{b}$ or
$\edge{b}{c}$, see Fig.~\ref{fig:ambig-quad}.
Thus $G'=G-C(a,b,c,d)-\edge{x}{y}$ decomposes into
$G_x$ and $G_y$ and there are 4-cliques $K(a,b,x,y)$ and
$K(b,c,x,y)$. Edges $\edge{a}{d}$ and
$\edge{c}{d}$ are uncrossed in $\Gamma(G)$, since $G'$ does not
decompose, otherwise. By 5-connectivity, all vertices have degree at
least five.

Suppose   that $G_y$ has  at least two vertices such that it   has a
vertex   $z \neq y$. We claim that
$C(a,y,c,d)$ is an ambiguous separating quadrangle such that
$\edge{a}{y}$ or $\edge{c}{y}$ is crossed by an edge that is
incident to $b$.
Assume that  $y$ and $z$ are outside $C(a,y,c,d)$. The  case where $z$ is inside $C(a,b,c,y)$ is
similar, using the paths between $z$  and $d$  instead of $z$ and
$b$.
There are five vertex disjoint paths between $z$ and $b$ in $G$ by
5-connectivity. Since edges $\edge{d}{c}$ and $\edge{d}{a}$ are
uncrossed, there is a path $\pi$ between $z$ and $b$ that crosses
$\edge{a}{y}$ or $\edge{c}{y}$. In fact, exactly one of these edges
is crossed, in alternation with $\edge{a}{b}$ or $\edge{c}{b}$ that
is crossed by $\edge{x}{y}$. That is, edge $\edge{x}{y}$ crosses
$\edge{v}{b}$ if and only if edge $\edge{v'}{y}$ is crossed by
$\pi$, such that $\{v, v'\}=\{a,c\}$. Assume that $\edge{c}{b}$ is
crossed. The case for $\edge{a}{b}$ is symmetric exchanging   $a$
and $c$. Then the first edge $\edge{b}{w}$ of $\pi$  crosses
$\edge{a}{y}$. Otherwise,  $C(a,b,y)$ is a separating triangle if
$\edge{a}{b}$ is uncrossed, and a separating triangle  if
$\edge{a}{b}$ and $\edge{a}{y}$ are crossed. In either case,
5-cliques containing the vertices from the pairs of crossed edges
are excluded by the uncrossed edges $\edge{a}{d}, \edge{c}{d},
\edge{a}{b}$ and $\edge{b}{y}$ and the fact that $\edge{x}{y}$ is
crossed. In particular, vertices $x$ and $b$ cannot be the fifth
vertex in such a 5-clique.

  Then $C(a,y,c,d)$ is  a separating
quadrangle with bridge $\edge{b}{w}$ or there is  edge $\edge{d}{y}$,
which is a chord, see Fig.~\ref{fig:quad-ambig2}. Towards a contradiction, suppose $\edge{d}{y}$
exits. Then $C(a,d,y)$ is a separating triple with bridge
$\edge{b}{w}$ if $\edge{d}{y}$ is uncrossed and a separating
triangle, if it is crossed. If the crossing edge is $\edge{u}{u'}$,
then there is no 5-clique containing $d,y,u,u'$, since e.g., $u=c$
and a 5-clique $G[d,y,c,u',a]$ implies that $C(a,b,c,d)$ has a
chord, which contradicts the assumption that it is an induced
4-cycle. By minimality, there is no unique separating
quadrangle. Hence,
   $C(a,y,c,d)$ is  a separating
quadrangle with bridge $\edge{b}{w}$, such that edge
$e'=\edge{c}{w}$ exists.
 It partitions into parts
$G_y-y$ and $G_x+b$, see Fig.~\ref{fig:quad-ambig2}. Now vertices
$a$ and $b$ are poles, that is neighbors of $d, x,b,y$ and $w$.

Let $|G_x| =i$ and suppose that $G$ has $k+2$ vertices. Define
$d=u_1$, $x=u_{i+1}$, $b=u_{i+2}$, $y= u_{i+3}$ and $w= u_{i+4}$. We now
proceed by induction and consider ambiguous separating quadrangles
with 4-cycle $C(a,u_j,c, u_1)$ for $j=i+2, \ldots, r$, where $r=k-2$ or $r=k-1$. We cannot
proceed if there is an edge $\edge{u_1}{u_{r+1}}$, that is, $C(a,u_{r+1},c,
u_1)$ has a chord. Then it turns out that there is an ambiguous separating triangle.
Next, consider part $G_y$ and proceed downwards
from $d=u_1, x=u_{i-1}, b=u_{i-2}, y= u_{i-3}$ and $w= u_{i-4}$.
Then there is a ambiguous separating triangle $C(a, u_{m-1},b,a_1)$ if there is
an ambiguous separating triangle $C(a, u_{m},b,a_1)$ and there is no chord
$edge{a_1}{a_{m-1}}$  for $m=i,i-1,\ldots, 2$. Let $u_l$ be the least vertex in
$G_y$ and $u_r$ the last vertex in $G_x$ such that there is an
ambiguous separating 4-cycle $C(a, u_l,c,u_1)$ and $C(a,
u_r,c,u_1)$, respectively, and there are edges $\edge{u_1}{u_{l-1}}$ and $\edge{a_1}{u_{r+1}}$

We claim that there is a circle $u_1,\ldots, u_k$, such that
$u_l=u_4$ and $u_r=u_{k-2}$ if there is a single vertex $u_1$ such
that edges $\edge{u_1}{a}$ and $\edge{u_1}{c}$ are uncrossed. If
there are two such vertices, then they are consecutive, that is also
$\edge{u_k}{a}$ and $\edge{u_k}{c}$ are uncrossed, since there
is a separating 4-cycle, otherwise. Then we use $C(a,u_k,b,a_m)$ for $G_y$
and have $u_l=u_3$.

There is a cycle  $u_1, u_{l-1}, u_l,\ldots, u_i, \ldots, u_r, u_{r+1}, u_k$,
since  the induction on $G_x$ stops at some
ambiguous separating quadrangle $C(a,u_1,c,u_r)$ and the induction on $G_y$
at $C(a,u_1,c,u_r)$. There is a
vertex $u_{r-1}$ such that edge $\edge{u_{r-1}}{u_{r+1}}$ crosses
$\edge{a}{u_r}$ or $\edge{c}{u_r}$ and there are edges
$\edge{a}{u_{r+1}}$ and $\edge{c}{u_{r+1}}$ and there is an edge
$\edge{u_{r+1}}{u_1}$, since $C(a,u_1,c, u_{r+1})$ has a chord $\edge{u_1}{u_{r+1}}$.
The case for $G_y$ is similar.

We claim that $u_r=u_{k-2}$ and $u_l=u_4$, such that two vertices  $u_2$ and $u_k$
are missing. By  the absence of
ambiguous separating quadrangles and smaller separators, subgraph
$G_x$ of the ambiguous separating quadrangle $C(a,u_1,c, u_l)$ is in
the interior of the 4-cycle, and similarly for $C(a,u_1,c, u_s)$.
Hence, we have a situation as depicted in
Fig.~\ref{fig:createcG2Sb}.
  Then consider $C(a,u_j, c,u_k)$
for (downwards) $j=i+2,\ldots, 3$. Thus the parts have at least
three vertices, that is there is no chord between $u_1$ and $u_j$
and $u_{k}$ and $u_j$, respectively.

There are four triangles $C(a, u_1, u_3), C(b, u_1, u_3), C(a, u_1, u_{k-1})$
and $C(b, u_1, u_{k-1})$.  The missing vertices $u_2$ and $u_{k-1}$ have degree at least
five by 5-connectivity. Their addition does not create a 5-clique.
Then $u_2$ can be placed in one of $C(a, u_1, u_3)$ or $C(b, u_1, u_3)$,
such that its neighbors are $a,b,u_1, u_3, u_4$. There is no edge $\edge{a_2}{a_k}$,
since edges $\edge{a}{u_1}$ and $\edge{b}{u_1}$ are uncrossed (from the
ambiguous separating quadrangles). Similarly, $u_k$ can be placed in one
of $C(a, u_1, u_{k-1}$ or $C(b, u_1, u_{k-1})$,
such that its neighbors are $a,b,u_1, u_{k-1}$ and $u_{k-3}$. Then $G = sG2_k$.

Suppose there is another ambiguous separating quadrangle $C(a, u, b, u_i)$ for
$u \neq u_1$. Then $u=u_k$ or $u=u_2$, since edges $\edge{a}{u}$ and $\edge{b}{u}$
are uncrossed, such that $C(a,u_1,b,u)$ is a separating 4-cycle, otherwise. Thus $u$ and $u_1$
are adjacent. Let $u=u_k$. Then
 $C(a, u_{k-1},c, u_1)$  is a separating 4-cycle with bridge $\edge{u_k}{u_{k-2}}$,
such that $G=cG2S_k$.
\qed
\end{proof}

A $cG2S_k$   usually occurs as a component at  a separating 4-cycle
$C(p,u_1,q,u_k)$, where edge $\edge{u_1}{u_k}$ is added for a
triangulation. The 4-cycle may result from another generalized
separator, for example a tripod (which is discussed below). Graph
$cG2S_k$ may have marked edges from the use of generalized
separators at an earlier stage. Then either the edges of a triangle
a marked or an edge incident to $u_1$ or $u_{k-1}$ has been removed.
 In the first case, there are two \emph{T1P} embeddings if all edges
$\edge{a}{u_i}$ and $\edge{c}{u_i}$ for $i=2,\ldots, k-1$ are
unmarked. There is unique embedding if an edge   $\edge{a}{u_i}$ or
$\edge{c}{u_i}$ with $2 \leq i \leq k-1$ is marked.
 There may be an error by  a wrong marking.

A $sG2S_k$ has at most two \emph{T1P} embeddings with the poles exchanged.
Algorithm $\mathcal{N}$ labels one edge of $cG2S_k$ or $sG2S_k$ by the number of
embeddings. It may remove $cG2S_k$ in a single sweep or remove all
edges between a vertex and it neighbor after next arches towards planarity. There may be an pre-check whether $G=sG2S_k$ without using
generalized separators, and similarly for $G=fG2S_k$.
The updates are as usual.

\subsection{Separating 5-Cycles} \label{sect:5cycle}
Similar to separating 4-cycles, we now   search for
 separating 5-cycles $C$, mark the edges of $C$, update
  pairs $(e, \mathcal{E}(e))$ and add two adjacent marked chords
 for the triangulation to each of the two remaining 6-connected components.
Observe that  $K_7$ is not 1-planar and 1-planar
graphs cannot be 8-connected, but there are 7-connected
1-planar graphs \cite{fm-s1pg-07}.

\subsection{Separating Tripods}

Forthcoming we can assume that $G$ is 6-connected.
If $G$ is a \emph{T1P} graph, then it has no 5-clique and no
separating triple or quadruple.
Now a 4-clique can be represented as a
tetrahedron if either  each edge of the tetrahedron is crossed by
some edge   and the vertices from the crossing edges
 are pairwise distinct, as shown in
Fig.~\ref{fig:kite-covered-tetrahedron},
 or  there is a vertex $d$ of degree six in the center of the
tetrahedron such that there are three 4-cliques, as shown in
Fig.~\ref{fig:SC-graph}. Chen et al.~\cite{cgp-rh4mg-06} call it
finding a rice-ball (section 9.1). In \cite{b-4mapGraphs-19} the
graphs are called complete kite-covered tetrahedron and SC-graph,
respectively. Here we continue in the line of separators and  use
tripods, that is 3-cycles with three bridges, since, again,
there may be an ambiguity.

\begin{lemma} \label{lem:tripod}  
If $C(v_1,v_2,v_3)$ is a separating tripod, such that
$G'=G-C(v_1,v_2,v_3) -
\mathcal{E}\edge{v_1}{v_2}-\mathcal{E}\edge{v_2}{v_3}-\mathcal{E}\edge{v_1}{v_3}$
is disconnected and each part contains at least two vertices, then
 one vertex of each bridge is inside $C(a,b,c)$ and the
other is outside,
  there is a single bridge $e_i= \edge{x_i}{y_i} \in
\mathcal{E}\edge{v_i}{v_{i+1}}$ for $i=1,2,3$ with $v_4=v_1$ such
that $G''=G-C(v_1,v_2,v_3) - \{e_1, e_2, e_3\}$ is disconnected if
$G'$ is disconnected.
 Graph $G'$ has   two parts $G_x$ and $G_y$, which
are only connected by the bridges $e_i$ for $i=1,2,3$.
 Edge $\edge{v_i}{v_{i+1}}$ is crossed by its bridge
$\edge{x_i}{y_i}$.
\end{lemma}

\begin{proof}
Since there are no 5-cliques, $G'$ does not decompose if the
vertices of an edge $e \in \mathcal{E}\edge{v_i}{v_{i+1}}$ (with
$v_4=v_1)$ are both inside (or outside) $C(a,b,c)$, as elaborated in
the proof of Lemmas~\ref{lem:single-bridge-lemma},
\ref{lem:single-bridge-lemma-quad} and
\ref{septriangle-twobridges}.    Hence, each bridge crosses an edge
of $C(a,b,c)$. Since $G$ is 6-connected, three
bridges are needed do decompose $G$ and each edge of $C(v_1, v_2,
v_3)$ is crossed at most once, such that the claims are clear.
\qed
\end{proof}

\begin{lemma}  \label{lem:distinct or strong tripod}  
Suppose  $C(v_1, v_2, v_3)$ is  a separating   tripod with bridges
 $\edge{x_i}{y_i}$ for $i=1,2,3$.
 Then either the edges $\edge{x_i}{y_i}$, $i=1,2,3$ are
pairwise independent, that is vertices $x_i, y_i$ for $i=1,2,3$ are
pairwise distinct or $x_1=x_2=x_3$ and $y_i \neq y_j$ for $1\leq i <
j \leq 3$. In other words, there is a completely kite-covered
tetrahedron. 
\end{lemma}
\begin{proof}
If $x_1=x_2$, then  $y_1 \neq y_2$ since there are no multi-edges.
Now there is a separating 4-cycle $C(v_1, x_2, v_2, x_3)$ or $C(v_1,
x_2, v_2, y_3)$, contradicting the assumptions if $x_3, y_3  \neq
x_1$. Hence all vertices $x_i, y_i, i=1,2,3$ are distinct or
$x_1=x_2=x_3$ and $y_1, y_2, y_3$ are pairwise distinct.
 \qed
\end{proof}

Algorithm $\mathcal{N}$ destroys a  separating  tripod with
independent bridges by removing the bridges.  The updates are as usual.
  There is a clear decision and a
single embedding.

A  separating tripod with independent edges is a piece of a
completely kite-covered tetrahedron, see
Fig.~\ref{fig:kite-covered-tetrahedron}. A completely kite-covered
tetrahedron has four separating tripods. For any of them, if its
bridges are removed, then there are unambiguous separating triangles
or separating triples, such that all  six X-quadrangles from the
kite-covered tetrahedron are detected and destroyed. Clearly, this
can be done all at once, as in \cite{cgp-rh4mg-06} or
\cite{b-4mapGraphs-19}.

\subsection{Strong Tripod Separators}

Now we use strong tripod separators to find its tetrahedron,
seeFig.~\ref{fig:SC-graph}.   Strong separating tripods play an
important role in the reduction system for optimal 1-planar graphs
\cite{b-ro1plt-18, s-s1pg-86}. They have a center $d$ of
degree six such that there are three 4-clique induced by $d$ and its
six neighbors. As before there is a
possible ambiguity that results in a full two-star $f2S_k$  that has
two embeddings (and eight for $k=6)$.

\begin{figure}[t]   
  \centering
\subfigure[] {
     \includegraphics[scale=0.4]{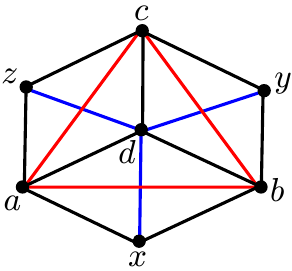}  
    \label{fig:tripod}
  }
\subfigure[] {
     \includegraphics[scale=0.4]{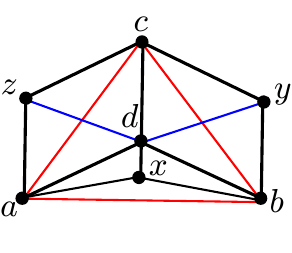}  
    \label{fig:trip1}
  }
  \hfil
  \subfigure[] {
      \includegraphics[scale=0.4]{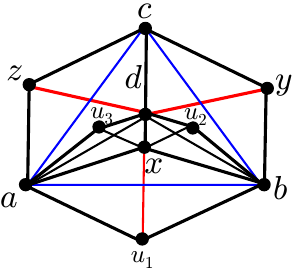} 
      \label{fig:trip2}
  }
\hfil
\subfigure[] {
     \includegraphics[scale=0.4]{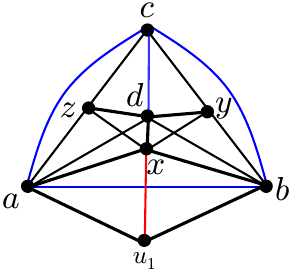}  
    \label{fig:trip3}
  }
\hfil
  \subfigure[] {
      \includegraphics[scale=0.4]{trip2} 
      \label{fig:trip4}
  }
  \caption{   Illustration for the proof of
  Lemma~\ref{lem:single-bridge-strong tripod}.
  (a) A  strong separating tripod $C(a,b,c)$ (b) with  $K(a,b,d,x)$ drawn as a tetrahedron,
(c)  with center $x$, (d) after setting $u_2=y$ and $u_3=z$,
(e) a non-strong separating tripod $C(a,b,c)$.
  }
  \label{fig:strong-tripod}
\end{figure}

\begin{lemma}  \label{lem:single-bridge-strong tripod} 
Let $C(a,b,c)$ be a strong separating tripod with center $d$. Then
the vertices of each bridge are on opposite sides of $C(a,b,c)$.
\end{lemma}

\begin{proof}
Towards a contradiction, suppose that $\edge{d}{x}$ is a bridge of
$\edge{a}{b}$ that does not cross an edge of $C(a,b,c)$, see
Fig.~\ref{fig:trip1}. Then $x$ is in the interior of $C(a,b,d)$.
The case where $x$ is in the interior of $C(b,c,d)$ or $C(c,a,d)$ is
similar. By the absence of smaller separators, any 3-cycle with vertices on either
side is a strong separating tripod. Hence, as $x$ is in and $c$ is outside $C(a,b,d)$,
it is a strong separating tripod with center $x$, such that its edges are crossed by
edges $\edge{x}{u_1}, \edge{x}{u_2}, \edge{x}{u_3}$, see Fig.~\ref{fig:trip2}.
Then $u_2=y$. since $\{b,c\}$ is a 2-separator, otherwise, contradicting 6-connectivity.
By a similar reasoning, we obtain $u_3=z$  see Fig.~\ref{fig:trip3}.
Then there is edge $\edge{y}{z}$, such that it crosses $\edge{c}{d}$, since
there are smaller separators, otherwise, that isolate $y$ and $z$, respectively.
 By 6-connectivity and the absence of smaller separators,
$C(b,c,d)$ is
a strong separating tripod  with center $y$ and   $C(a,c,d)$ is a   strong
separating tripod  with center $z$, as shown in
Fig.~\ref{fig:trip4}. Hence $C(a,b,c)$ is a separating tripod, that
is not strong, a contradiction to the assumption.
%
%
\qed
\end{proof}

\begin{lemma} \label{strong-unambiguous-SC-if}  
Let $d$ be the center of a strong  separating tripod with 3-cycle
$C(a,b,c)$ and bridges $\edge{d}{x}, \edge{d}{y}$  and $\edge{d}{z}$.
There is another   strong  separating tripod with 3-cycle $C(b,c,y)$
and bridges $\edge{d}{a}, \edge{d}{z}$ and $\edge{d}{x}$ if $G$ contains the
  edges $\edge{x}{y}$ and $\edge{y}{z}$.
\end{lemma}
\begin{proof}
There may be two embeddings $\Gamma_1(G)$ and $\Gamma_2(G)$ such
that the cyclic order of the neighbors of $d$ is $(a,x,b,y,c,z)$ and
$(a,b,x,y,z,c)$, respectively, since there are 4-cliques if two
edges cross, see Fig.\ref{fig:strong-SC}
\qed
\end{proof}

\begin{figure}[t]  
  \centering
\subfigure[] {
     \includegraphics[scale=0.5]{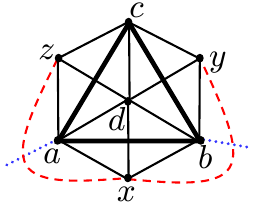}  
    \label{fig:tripod1}
  }
  \hfil
  \subfigure[] {
      \includegraphics[scale=0.5]{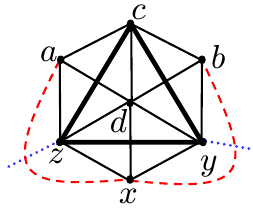} 
      \label{fig:tripod2}
      }
  \caption{Illustration for the proof of
  Lemma~\ref{strong-unambiguous-SC-if}.
   An ambiguous separating tripod $C(a,b,c)$ with center $d$ and
  edges $\edge{x}{y}$ and $\edge{x}{z}$ that admits two \emph{T1P}
  embeddings. Vertices $c$ and $x$ are poles that are neighbors of
  all other vertices.
  }
  \label{fig:strong-SC}
\end{figure}

\begin{lemma} \label{lem:only-wheelgraphs}
If a    T1P  graph $G$ has an ambiguous separating tripod
and there are no smaller separators, then $G$   is the full two-star $f2S_k$.
\end{lemma}

\begin{proof}
Let $C(a,b,c)$ be an ambiguous separating tripod with bridges
$T=\{\edge{x_i}{y_i} \, | \, i=1,2,3\}$, such that there is a center
$d$. Then there are six 4-cliques. We claim that there is an
ambiguous separating tripod with center $a$ and one with center $b$,
where $a$ and $b$ have degree six. Otherwise, let $w$ be the sixth
neighbor of $a$. If there are 4-cliques $G[w,a,x,z]$ and
$G[w,a,c,z]$, then $a$  is the center of an ambiguous separating
tripod. Otherwise, if, for example, $G[w,a,x,z]$ does not exist,
then one of $\edge{w}{x}$ or $\edge{w}{z}$ is missing. Then there is
a strong separating triangle $C(d,x,z)$ with bridges $\edge{a}{b}$
and $\edge{a}{c}$, contradicting the minimality of $C(a,b,c)$. By
induction, there is a sequence   of centers of ambiguous separating
tripods $v_1, v_2, \ldots, v_k$ with $v_1=d$ and $v_2=a$.  Then
$v_k=b$, since there is a strong separating triangle, otherwise.
Hence, there is a cycle $v_1, v_2, \ldots, v_k$ and there are two
poles $p=c$ and $q=x$, such that $G=fG2S_k$.
 \qed
\end{proof}

Observe that a 1-planar embedding of a full generalized two star $fG2S_k$ has
only two vertices in its outer boundary, and is has no triangle with
uncrossed edges. Hence, it cannot occur as part of a larger
3-connected \emph{T1P} graph and it does not have a marked edge. It
is known that $f2S_6$, that is the crossed cube, has eight 1-planar
embeddings, whereas all other full two stars have two 1-planar
embeddings \cite{s-rm1pg-10}. Clearly, full generalized two stars can
directly be
recognized in linear time using the degree of the vertices. If it is
recognized by algorithm $\mathcal{N}$, then a strong separating
tripod it found first, such that bridges $\edge{d}{y}$ and
$\edge{d}{z}$ are removed, where $y$ and $z$ (and $d)$ have degree
six and are not poles (for $G \neq f2S_6$).
This completes the description of step (4).

As observed by Chen et al.~\cite{cgp-rh4mg-06}, for any of the
remaining 4-cliques, either all edges are marked, such that that $K_4$
is represented as a tetrahedron, or it has a pair of unmarked edges
$e, e\rq{}$ with $e \in  \mathcal{E}(e\rq{})$ and $e\rq{} \in  \mathcal{E}(e)$
In the latter case, algorithm $\mathcal{N}$ removes $e$ and $e\rq{}$ and
replaces them by a 4-star, as described before, and then runs
a planarity test  for step (7) and checks
whether each component is a triangulated planar graph.
   It succeeds if and only if
the planarity test is successful.\\

We summarize for our main result.

\begin{proof} of Theorem~\ref{thm1}.

By induction on the used generalized separators, if graph $G$ is a
\emph{T1P} graph, then algorithm $\mathcal{N}$ decomposes $G$
by separating $k$-cycles. Each of the other generalized separators finds
an edge in a 4-clique that must be crossed in an X-quadrangle.
All   4-cliques that must be represented
as an X-quadrangle are destroyed by removing one crossed edge,
such that the  final planarity test will succeed.
Conversely, if algorithm $\mathcal{N}$ succeeds, then
a \emph{T1P} embedding of $G$
is reconstructed as a witness,
where crossed edges are re-inserted such that there is an X-quadrangle,
and embeddings are composed
at $k$-separators.
In addition, $\mathcal{N}$  computes the edge labels
for an ambiguity of an embedding at each generalized separator,
where small graphs are searched by  exhaustive search, such that
the number of \emph{T1P} embeddings is the product of the labels.

Concerning the running time, we argue as in \cite{cgp-rh4mg-06,
b-4mapGraphs-19}. An n-vertex graph has
$O(n)$ maximal 4- and 5-cliques, which together with all $O(n)$ $j$-cycles
for $j=3,4,5$ can be computed in linear time
by the algorithm Chiba and Nishizeki
\cite{cn-asla-85}. It takes $O(n)$ time to test whether $G$ decomposes
by any generalized separator, and there are at most $O(n)$
``make progress\rq{}\rq{} steps.
Hence, $\mathcal{N}$ runs in cubic time.
\qed
\end{proof}

\section{Applications} \label{sect:app}

So far, the \emph{T1P} embeddings of small graphs of order at most six
are computed by exhaustive search. The bound six is due to a need for
seven vertices in some proofs.
The small \emph{T1P}  graphs are  $l$-cliques for $l\leq 6$ and
  subgraphs that are obtained by removing
some edges and still contain  a triangulated planar graph.

 \begin{table}[h]
\centering
\begin{tabular}[h] {l |c       |c            |c                |c            |c          | c          |c             | c             | c                      |c            }
\emph{T1P} graph & $K_3$ & $K_4$ & $K_5$-$e$ & $K_5$ & $H_2$ &  $H_5$ & $H_6$ & $H_7$    & $K_6$-$e$  & $K_6$ \\
\hline
graph                   &       1    &  4      &    1             &   15      &    1      &    1       &     6         & 4             &    12     &   1 \\
outer face          &        1   &   1     &    1              &   $\leq 6$      &     1     &  1 &    $\leq 6$      & $\leq 3$          &    $\leq12$     &   1
\end{tabular}
\caption{All small \emph{T1P} graphs,  (first line)
the number of \emph{T1P} embeddings  for the graph and
with a fixed outer face   (second line)
.}
\label{tab:small}
\end{table}

\begin{figure}[t]  
  \centering
\subfigure[$K_5$-e   planar] {
     \includegraphics[scale=0.63]{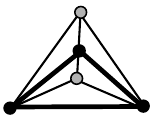}
    \label{fig:K5-e-T1P}
  }
  \subfigure[$K_5$-e \hspace{5mm} kite-aug.] {
      \includegraphics[scale=0.65]{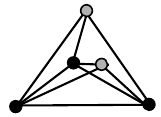} 
      \label{fig:K5-e-aug}
      }
 \subfigure[$K_5$-e 1-planar] {
     \includegraphics[scale=0.63]{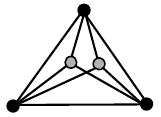}  
     \label{fig:K5-e-ill}
   }
\subfigure[$H_1$] {
     \includegraphics[scale=0.6]{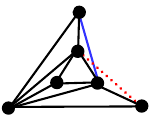}  
    \label{fig:small-H1}
  }
\subfigure[$H_2$] {
     \includegraphics[scale=0.6]{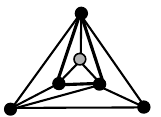}  
    \label{fig:small-H2}
  }
\subfigure[$H_3$] {
     \includegraphics[scale=0.6]{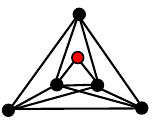}  
    \label{fig:small-H3}
  }\\
\subfigure[$H_4$] {
     \includegraphics[scale=0.6]{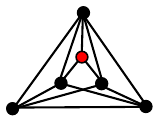}  
    \label{fig:small6-H_4}
}
\subfigure[$H_5$] {
     \includegraphics[scale=0.6]{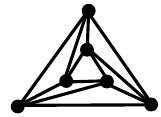}  
    \label{fig:small-H5}
  }
\subfigure[$H_6$] {
     \includegraphics[scale=0.5]{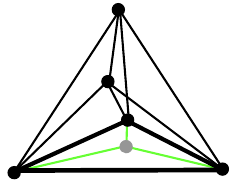}  
    \label{fig:H3}
  }
\hspace{3mm}
\subfigure[$H_7$] {
     \includegraphics[scale=0.6]{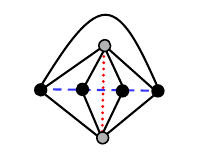}  
    \label{fig:small-H7}
  }
\hspace{3mm}
 \subfigure[$K_6$-e] {
     \includegraphics[scale=0.5]{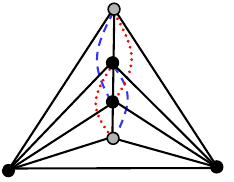}  
     \label{fig:K6-e}
   }
  \caption{Illustration for the proof of
  Lemma~\ref{lem:small-graphs}. All small \emph{T1P} graphs and 1-planar embeddings,
that are \emph{T1P}  or kite-augmented or just 1-planar (c).
(i) Graph $K_6$-e has either both red dotted or both blue dashed edges.
  }
  \label{fig:small}
\end{figure}

\begin{lemma} \label{lem:small-graphs}
Table~\ref{tab:small} lists all  small  T1P graphs  of order at most six, the number of   T1P embeddings, and
the number of    T1P embeddings if the outer face is fixed and its edges are uncrossed.
\end{lemma}
\begin{proof}
Clearly,
$K_3, K_4, K_5$-$e$ and $K_5$ are the \emph{T1P} graphs of order
at most five, as all other graphs of order at most five are not triangulated.
 The candidates for \emph{T1P} subgraphs of $K_6$
are obtained by removing one, two or three edges, where two removed
edges can be adjacent or independent, and three removed edges are
 a triangle for $H_1$,  a path for $H_2$,  a star for $H_3$,
   a path and an edge for $H_4$, and
three independent edges for $H_5$. The obtained graphs are distinguished
by the degree of their vertices.
The number of \emph{T1P} embeddings of $k$-cliques
with and without a fixed outer face
has been discussed before. Observe that
exactly one pair of edges must cross in a 1-planar
embedding of $K_5$~\cite{hm-dcgmnc-92}, and any pair of independent edges
may do.
 Graph $K_5$-$e$ is planar and has a separating 3-cycle,
such that three edges can be marked. Then only the planar embedding is a \emph{T1P} embedding, see Fig.~\ref{fig:K5-e-T1P}.
The 6-clique $K_6$ has a single 1-planar and \emph{T1P} embedding, respectively,
see Fig.~\ref{fig:K6-triangle}.

There are five subgraphs of $K_6$ with twelve edges, denoted
 $H_1,\ldots, H_5$, but only $H_2$ and $H_5$
are \emph{T1P} graphs, see  Fig.~\ref{fig:small}.
Graph $H_1$ has three vertices of degree three  and three of
degree five, and a $K_5$ minor. It is 1-planar, but  not a \emph{T1P} graph,
since the triangulation needs an edge between two degree-3 vertices.
Graph $H_2$ is planar. It has a separating
3-cycle, such that $K_5$-$e$ remains as one component with a fixed outer face.
Hence, $H_2$ has a single \emph{T1P} embedding, namely the planar one.
Graph $H_3$ has a vertex of degree two, since three incident edges
are removed. Thus it is not 3-connected and not a \emph{T1P} graph.
Similarly, graph $H_4$ is not a \emph{T1P} graph, since it has a
single vertex of degree three, whose neighbors do not induce a 3-cycle.
Graph $H_5$   is planar
and 4-regular and does not contain $K_4$. Hence, algorithm $\mathcal{N}$
marks all edges, such that only the planar embedding is a \emph{T1P} embedding.
Clearly, graphs $H_2$ and $H_5$ each have a single \emph{T1P} embedding
with any fixed outer face.

Two edges are either adjacent or independent, such that
graphs  $H_6$  and $H_7$ are the only \emph{T1P} graphs
with 6 vertices and 13 edges. Graph $H_6$,  shown in Fig.~\ref{fig:H3},
consists of a vertex of degree three
and a $K_5$  with two degree-4 vertices in the interior of the cycle
around the degree-3 vertex.
Hence, $H_6$  has  six \emph{T1P} embeddings that are inherited from $K_5$
with a fixed outer face.
Graph $H_7$,  shown in Fig.~\ref{fig:small-H7}, has a separating edge between the
two vertices of degree five, such that every edge
between two degree-4 vertices is a crossable edge.  Hence, it admits four  \emph{T1P} embeddings, and up to
three  with a fixed outer face.
It is also shown in Fig.~\ref{fig:quad-A} with an extension.
Clearly, the 6-clique without one edge is the only
 \emph{T1P}  subgraph of $K_6$ with 14 edges, see Fig.~\ref{fig:K6-e}.
Each of the two degree-4 vertices is incident to a crossed edge.
Any two vertices of degree five can be
taken as the poles of a (too small) generalized two-star $G2S_4$ that each admits
two \emph{T1P} embeddings, such that there are $2\binom{4}{2}$ \emph{T1P} embeddings. Alternatively, observe that any permutation of the edges
incident to a single degree-4 vertex creates a different \emph{T1P} embedding.
It  has  one $K_5$-triple if the outer cycle is fixed, that is  two \emph{T1P} embeddings.
\qed
\end{proof}

By Lemma~\ref{lem:small-graphs}  and Table~\ref{tab:small},
 only $K_5$ and the subgraphs
$H_6, H_7$ and $K_6$-$e$  of $K_6$  must be
avoided as small  graphs at separators if graphs shall be unique,
or  further edges must be marked, for example, if there is a
vertex (subgraph) in the interior of some triangles.
Hence, the exhaustive search can be simplified to these few cases.\\

We now consider the uniqueness of a \emph{T1P} graphs, which is
obtained by algorithm $\mathcal{N}$ and proved by its correctness.
In contrast, the planar case is proved
by using  graph theoretic and
combinatorial methods,   for example in \cite{d-gt-00}.

\begin{theorem} \label{thm:uniquetime}
There is a cubic time algorithm that decides whether or not a  T1P
graph  is unique.
\end{theorem}
\begin{proof}
By Theorem~\ref{thm1}, a \emph{T1P} graph has a single \emph{T1P}
embedding if algorithm $\mathcal{N}$ succeeds and each edge label is equal
to one. Otherwise, there is an ambiguous separator and at least two
\emph{T1P}  embeddings.
\qed
\end{proof}

For the uniqueness of a \emph{T1P} graph, we may simply ask
algorithm  $\mathcal{N}$. The answer will be as follows.

\begin{corollary} \label{cor:uniquecharcterize}
A  T1P  graph $G$  
is unique if and only if it does not have a separating edge or an ambiguous
separating triple, quadruple, triangle and tripod, respectively,
and it is  not or does not contain a small graph at a separating $k$-cycle
that admits two or more T1P embeddings as a part of $G$.
\end{corollary}

By Whitney\rq{}s theorem \cite{w-uniqueplanar-33},
 any 3-connected planar graph  admits a single planar embedding.
However,   it may admit many 1-planar ones. As an example, consider an $m \times n$
grid
with quadrangles in the interior
and edges between the corners, that can be folded in many ways for
a 1-planar embedding, or see Fig.~\ref{fig:small-planar}.
Surprisingly, the allowance of single edge crossings
does not yield more embeddings for triangulated graphs.

\begin{corollary} \label{cor:planar}
Every triangulated planar graph has a single  T1P  embedding.
\end{corollary}
\begin{proof}
If $G$ is a triangulated planar graph, then algorithm $\mathcal{N}$
 uses only  3- and 4-separators and the final planarity test,
such that $\lambda(e)=1$ for every edge.
\qed
\end{proof}

Recall that 2-connected planar graphs
may admit exponentially many
planar embeddings and there are \emph{T1P} graphs
with exponentially many \emph{T1P} embeddings~\cite{bbhnr-NIC-17}.\\

There are important graph classes between the planar and the 1-planar graphs.
A 1-planar embedding is \emph{IC-planar} if each vertex is incident to
at most one crossed edge, and \emph{NIC-planar} if each edge is in
at most one  X-quadrangles \cite{klm-bib-17}, see also
\cite{bdeklm-IC-16, b-IC+NIC-18,bbhnr-NIC-17}.
The recognition of triangulated \emph{IC}- and \emph{NIC-}planar graphs needs
a test for the specific restriction. This can be expressed by a boolean
formula, whose evaluation may take extra time or
enforce a restriction \cite{b-IC+NIC-18}, that vanish for uniqueness.
Also optimal \emph{RAC} graphs range between the triangulated planar and 1-planar
graphs~\cite{del-dgrac-11, el-racg1p-13}. However,
\emph{RAC} graphs are defined  by a straight-line drawing with
right angle crossings, which cannot be treated with our tools,
as \cite{s-RAC-drawability-23} suggests.

\begin{corollary} \label{cor:IN+NIC}
There is a cubic time algorithm that decides whether or not a  T1P
graph  has a  single  IC-planar and NIC-planar  embedding, respectively.
In particular, it is decided whether or not a triangulated IC- and NIC graph
is unique.
\end{corollary}
\begin{proof}
If algorithm $\mathcal{N}$ succeeds on graph $G$ with $\lambda(e)=1$
for every edge, then it returns the single \emph{T1P} embedding of $G$,
for which the \emph{IC} and \emph{NIC} property can be checked straightforwardly.
\qed
\end{proof}

Next, we generalize results by Schumacher~\cite{s-s1pg-86} and
Suzuki~\cite{s-rm1pg-10} from optimal 1-planar graphs, which have
the maximum of $4n-8$ edges
and thus are triangulated, to 6-connected 1-planar graphs.
Note that optimal 4-planar graphs are either 4-connected, in which
case they have separating 4-cycles, or they are 6-connected \cite{s-s1pg-86}.
 Clearly, there are
6-connected \emph{T1P} graphs that are not optimal, as
the graphs in Fig.~\ref{fig:dense}
show.

\begin{figure}[t]  
  \centering
\subfigure[] {
     \includegraphics[scale=0.5]{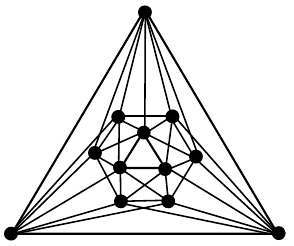}  
    \label{fig:non-optimal}
  }
  \hspace{5mm}
  \subfigure[] {
      \includegraphics[scale=0.48]{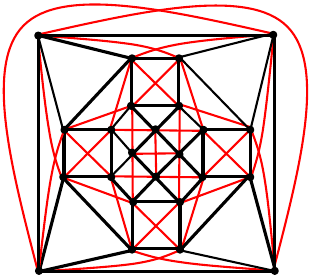} 
      \label{fig:7connected}
      }
  \caption{(a) A 6-connected 1-planar graph with 12 vertices
and 39 edges (not optimal)
and (b) a 7-connected 1-planar graph from \cite{fm-s1pg-07}
  }
  \label{fig:dense}
\end{figure}

\begin{corollary} \label{cor:uniqueNotwostar}
A 6-connected  T1P  graph $G$ 
has a  unique T1P  embedding,
except if $G$ is  a full
two-star, that is $G= fG2S_k$, that has two 1-planar
embeddings which are triangulated. 
\end{corollary}
\begin{proof}
Small graphs are not 6-connected. If $G$ is 6-connected, then
algorithm $\mathcal{N}$ admits only tripods.
Now $G$ has a single \emph{T1P} embedding or
 $G= fG2S_k$ if there is an ambiguous tripod, as shown in
Lemma~\ref{lem:only-wheelgraphs}.
Graph $fG2S_k$ has two 1-planar embeddings \cite{s-s1pg-86}.
\qed
\end{proof}

We wish to extend algorithm $\mathcal{N}$ to kite-augmented 1-planar
and (general) 4-map graphs, similar to \cite{b-4mapGraphs-19}.
This can be achieved for the recognition part. 
However, there are problems with counting the embeddings, since
a graph may have   kite-augmented 1-planar embeddings that are
not \emph{T1P},
as Fig.~\ref{fig:small-planar} shows.
Moreover, we would like to count the number of 1-planar embeddings
of \emph{T1P} graphs.  

\section{Conclusion}
\label{sect:conclusion}
We describe a cubic time recognition algorithm for \emph{T1P} graphs
that also counts the number of \emph{T1P} embeddings.
 Thereby, we
obtain an algorithmic solution for the uniqueness problem for
\emph{T1P} graphs and show that every triangulated planar graph
has a single \emph{T1P} embedding.
In passing, we  clean up an earlier recognition algorithm that uses
hole-free 4-map graphs of its presentation.

We would like to extend algorithm $\mathcal{N}$
to kite-augmented 1-planar graphs along the lines of \cite{b-4mapGraphs-19}.
Maps and witnesses are ambiguous and usually define an edge of a
map graph is several ways. What are the restrictions if the defined
(hole-free) 4-map graph is unique? At last, the recognition problem
for map graphs is open.



%
\end{document}